\documentclass[12pt]{article}
\usepackage[utf8]{inputenc}

\usepackage{graphicx}% Include figure files
\usepackage{dcolumn}% Align table columns on decimal point
\usepackage{bm}% bold math
\usepackage{hyperref}% add hypertext capabilities
\usepackage[mathlines]{lineno}% Enable numbering of text and display math
\usepackage[a4paper,margin=2cm]{geometry}

\usepackage{amsmath}
\usepackage{amssymb}
\usepackage{amsthm}
\usepackage{mathtools}
\usepackage{bbold}
\newcommand{\RNum}[1]{\uppercase\expandafter{\romannumeral #1\relax}}

\usepackage{authblk}

\usepackage{qcircuit}
\usepackage{physics}
\usepackage{graphicx}
\usepackage{caption}
\usepackage{subcaption}
\usepackage{cite}
\theoremstyle{definition}
\newtheorem{theorem}{Theorem}
\newtheorem{lemma}{Lemma}

\newcommand{\snorm}[1]{\left|\left|#1\right|\right|_{1\rightarrow1}}
\newcommand{\Lk}[1]{\hat{\mathcal{L}}_{#1}}

\usepackage{hyperref}

\title{Digital Simulation of Single Qubit Markovian Open Quantum Systems: A Tutorial}

\date{}

\author[1]{I. J. David}

\author[1,2]{I. Sinayskiy}

\author[1,2,3]{F. Petruccione}

\affil[1]{School of Chemistry and Physics, University of KwaZulu-Natal, Durban, South Africa.}

\affil[2]{National Institute for Theoretical and Computational Sciences (NITheCS), South Africa.}

\affil[3]{School of Data Science and Computational Thinking, Stellenbosch University, Stellenbosch 7604, South Africa.}

\begin{document}

\maketitle

\section*{Abstract}
{\bf

One of the first proposals for the use of quantum computers was the simulation of quantum systems. Over the past three decades, great strides have been made in the development of algorithms for simulating closed quantum systems and the more complex open quantum systems. In this tutorial, we introduce the methods used in the simulation of single qubit Markovian open quantum systems. It combines various existing notations into a common framework that can be extended to more complex open system simulation problems. The only currently available algorithm for the digital simulation of single qubit open quantum systems is discussed in detail. A modification to the implementation of the simpler channels is made that removes the need for classical random sampling, thus making the modified algorithm a strictly quantum algorithm. The modified algorithm makes use of quantum forking to implement the simpler channels that approximate the total channel. This circumvents the need for quantum circuits with a large number of C-NOT gates. 
}

\tableofcontents

\section{Introduction}

The simulation of complex quantum systems is central to many important problems in drug discovery, chemistry, material science and more. Although simulating these systems remains a computationally challenging task for classical computers, the natural resources of quantum computers can be leveraged to more efficiently simulate these systems.

Generally, simulating a quantum system involves approximating a mathematical object describing the total evolution of a system by some combination of simpler objects. When the simulation is performed on a quantum computer, these simpler objects define the operations that will be performed on the quantum computer. Over the past three decades, great strides have been made in the development of algorithms for simulating closed and open quantum systems. These algorithms often strive to achieve a low enough error in the approximation of the mathematical object and find an efficient set of operations to perform the evolution.

 The simulation of closed quantum systems, often called Hamiltonian Simulation, has been extensively explored \cite{lloyd1996universal,campbell2019random,childs2012hamiltonian,childs2018toward,childs2019faster,childs2021theory,berry2007efficient,berry2015simulating,berry2015hamiltonian,berry2020time,papageorgiou2012efficiency,low2019hamiltonian}. For Hamiltonian Simulation, the mathematical object that must be approximated is the total unitary evolution generated by some Hamiltonian. The total unitary evolution is approximated by some combination of simpler unitaries, generated by the components of the decomposed Hamiltonian. Once the simpler unitaries are obtained, efficient quantum circuits for these simpler unitaries which minimise gate count and number of qubits must be constructed. Algorithms for Hamiltonian simulation are typically characterised by the recombination methods used to combine the simpler unitaries. The most commonly used recombination method is the Suzuki-Lie-Trotter (SLT) Product Formulas \cite{suzuki1990fractal,suzuki1991general}.
 
While closed quantum systems are relevant to many problems in quantum mechanics, open quantum systems describe a more complex, commonly encountered system: a quantum system that can interact with its surrounding environment, sharing information and energy \cite{breuer2002theory,rivas2012open}. These systems are crucial to our understanding of nonequilibrium dynamics and thermalization in a wide range of systems, from damped-driven spin-boson models to complex many fermion-boson models \cite{georgescu2014quantum,brown2010using}. In this tutorial, we shall only consider the simulation of only Markovian open quantum systems for which we can neglect memory effects.

Several methods that allow us to simulate open quantum systems have been developed  \cite{sweke2014simulation,bacon2001universal,sweke2015universal,childs2016efficient,cleve2016efficient,hu2020quantum,gaikwad2022simulating,kamakari2022digital,suri2022two}. For the simulation of open quantum systems, the mathematical object that must be approximated is a quantum channel: a completely positive and trace preserving (CPTP) map, also called a dynamical map. Like the Hamiltonian which generates the total unitary evolution for closed quantum systems, a superoperator called the Gorini-Kossakowski-Sudarshan-Lindblad (GKSL) generator \cite{lindblad1976generators,gorini1976completely} generates the quantum channel. The quantum channel can then be approximated by a combination of simpler channels generated by the components of the decomposed GKSL generator. These simpler channels are then executed on the quantum computer.

While the spectral decomposition was the natural choice for decomposing the Hamiltonians in Hamiltonian simulation, the method used to decompose the GKSL generator for open system simulation was developed by Bacon et al. \cite{bacon2001universal}. This method involves decomposing the GKSL generator into simpler components via linear combination of semi-groups and unitary conjugation. It was also shown by Bacon et al. that all quantum channels on finite dimensional systems can be simulated by unitary evolutions and quantum channels from a universal semigroup. The decomposition method developed was used by Sweke et al. to develop an algorithm for simulating open quantum systems \cite{sweke2014simulation,sweke2015universal}. This algorithm makes use of SLT product formulas for the recombination of the simpler quantum channels. Sweke et al. also derived the bounds on the error in the approximation as well as the bounds on the gate counts.

Despite the advancements made in the simulation of open systems, a standard or widely accepted notation is lacking. This has hindered progress in the field. This tutorial will introduce the methods used in the simulation of single qubit open quantum systems. It combines various existing notations into a common framework that can be extended to more complex open system simulation problems. The only currently available algorithm for the digital simulation of single qubit open quantum systems \cite{sweke2014simulation} will be discussed in detail. A modification that makes use of quantum forking \cite{park2019parallel} in the implementation of the simpler channels is made that removes the need for classical random sampling. This makes the modified algorithm a strictly quantum algorithm. This modification also circumvents the need for quantum circuits with a large number of C-NOT gates.

 This tutorial will be structured as follows: Section 2 formally states the problem of quantum simulation and fixes the notation used throughout this tutorial. Section 3 and 4 present the decomposition and recombination methods used by the algorithm, respectively. Section 5 presents the quantum circuits that can implement any channel from the universal semigroup of quantum channels. Lastly in Section 6, concluding remarks are made and open problems in open quantum system simulation are discussed. 

\section{Problem and Setting}
\label{background&problem}
The state space of a single qubit is a  two dimensional complex Hilbert space $\mathcal{H}_{s} \cong \mathbb{C}^{2}$. A quantum state of a single qubit is described by a density matrix $\rho \in \mathcal{M}_{2}(\mathbb{C}) \cong \mathcal{B}(\mathcal{H}_{s})$, where $\mathcal{B}(\mathcal{H}_{s})$ is the set of all bounded linear operators acting on the Hilbert space $\mathcal{H}_{s}$ and $\mathcal{M}_{2}(\mathbb{C})$ is the set of $2\times 2$ complex matrices , such that:
\begin{align}
\label{eq1}
	 \rho \geq 0,  && \mathrm{tr}(\rho)=1, && \rho=\rho^{\dagger}.
\end{align}
Quantum channels provide a general framework for describing the evolution of quantum states. These are completely positive and trace preserving (CPTP) maps \cite{breuer2002theory},
\begin{equation}
\label{eq2}
    T:\mathcal{B}(\mathcal{H}_{s}) \rightarrow \mathcal{B}(\mathcal{H}_{s}).
\end{equation}
It is a well known fact that $T$ has a Kraus representation \cite{kraus1971general},
\begin{equation}
\label{eq3}
    T(\rho)=\sum_{j=1}^{r}K_{j}\rho K_{j}^{\dagger},
\end{equation}
where $K_{j}$ are called the Kraus operators that satisfy the completeness relation $\sum_{j=1}^{r}K_{j}^{\dagger}K_{j}=\mathbb{1}_{s}$, with $\mathbb{1}_{s}$ being the identity operator acting on the space $\mathcal{H}_{s}$, and $r=\mathrm{rank}(\tau)\leq 4$. Here $\tau$ is the Choi matrix defined through the the Choi-Jamiolkowski isomorphism \cite{choi1975completely,jamiolkowski1972linear} as,\\
\begin{equation}
\label{eq4}
    \tau=(T \otimes \mathbb{1}_{s})|\Omega\rangle \langle \Omega|,
\end{equation}
where $|\Omega \rangle = \frac{1}{\sqrt{2}}(\ket{00}+\ket{11})$ is a maximally entangled state. The Choi matrix is a representation of the quantum channel $T$ and it is defined through the Choi-Jamiolkowski isomorphism. It is a matrix description of the quantum channel $T$ and it contains all the properties of $T$. For a quantum channel $T$, we can define the dual or adjoint as $\tilde{T}$ and it can be calculated through the relation $\mathrm{tr}(A\tilde{T}(B))=\mathrm{tr}(T(A)B)$ for any $A,B \in \mathcal{B}(\mathcal{H}_{s})$. Using this relation we can also find the Kraus representation of the dual of the channel,
\begin{equation}
\label{eq5}
	\tilde{T}(\rho)=\sum_{j=1}^{r}K_{j}^{\dagger}\rho K_{j}.
\end{equation}

So far we have only considered the state of a single qubit $\rho$ that is time independent, but for most practical cases we must consider a time dependent state $\rho(t)$ for some $t\geq 0$.

The channels above describe discrete time evolution as the channel $T$ does not depend on some continuous parameter and are specified by time independent Kraus operators. However we are interested in Markovian continuous time evolution. Which is described by a continuous single parameter semigroup of quantum channels $\{T_{t}\}$, which satisfy:\\
\begin{equation}
\label{eq6}
    T_{t}T_{s}=T_{t+s}, \hspace{10mm} T_{0}=\mathbb{1} \hspace{10mm} t,s \in \mathbb{R}_{+}.
\end{equation}
Also if we introduce time dependence to the state of our single qubit, then we write the density matrix, $\rho(t)$ which describes the quantum state at some time $t \geq 0$ then we have that,\\
\begin{equation}
\label{eq7}
    \rho(t)=T_{t}(\rho(0)).
\end{equation}
Every semigroup $\{T_{t}\}$ has a unique generator $\mathcal{L}: \mathcal{B}(\mathcal{H}_{s}) \rightarrow \mathcal{B}(\mathcal{H}_{s})$ such that,\\
\begin{equation}
\label{eq8}
    T_{t}=e^{t\mathcal{L}}=\sum_{j=0}^{\infty}\frac{t^{j}}{j!}\mathcal{L}^{j},
\end{equation}
where $\mathcal{L}$ satisfies the master equation,\\
\begin{equation}
\label{eq9}
    \frac{d}{dt}\rho(t)=\mathcal{L}(\rho(t)).
\end{equation}
The generator $\mathcal{L}$ is the generator of a continuous one parameter Markovian semigroup $\{T_{t}\}$ if and only if it can be written in the celebrated Gorrini-Kossakowski-Sudarshan-Lindblad (GKSL) form \cite{gorini1976completely,lindblad1976generators},\\
\begin{equation}
\label{eq10}
    \mathcal{L}(\rho)= -i[H,\rho] +\sum_{i,j=1}^{3} A_{ij} \big([F_{i},\rho F_{j}^{\dagger}]+[F_{i}\rho ,F_{j}^{\dagger}]  \big),
\end{equation}
where $H=H^{\dagger} \in \mathcal{M}_{2}(\mathbb{C})$ is the Hamiltonian and $A \in \mathcal{M}_{3}(\mathbb{C})$ is a positive semi-definite matrix called the GKS matrix. The matrices $\{F_{i}\}$ are a basis for the space of traceless matrices in $\mathcal{M}_{2}(\mathbb{C})$ and without loss of generality we chose the basis $\{F_{i}\}$ to the be the normalized Pauli matrices i.e. $\{F_{i}\}=\frac{1}{\sqrt{2}}\{\sigma_{i}\}_{i=1}^{3}=\frac{1}{\sqrt{2}}\{ \sigma_{1}, \sigma_{2}, \sigma_{3}\}$.\\
\\
One of the main aspects of digital quantum simulation of an open quantum system is finding an approximation to the quantum channels that describe the evolution of the system. 
 To quantify the error in approximations of quantum channels we make use of the $1\rightarrow1$ superoperator norm defined in terms of the general $p\rightarrow q$ Schatten norm. One should note quantum channels are also called superoperators since they are operators which act on other operators. The general $p\rightarrow q$ Schatten norm for some superoperator $T$ is defined as:\\
\begin{equation}
\label{eq11}
    ||T||_{p \rightarrow q}= \sup_{||A||_{p}=1}||T(A)||_{q}
\end{equation}
where $||A||_{p}:= (\mathrm{tr}(|A|^{p}))^{\frac{1}{p}}$ and $|A|=\sqrt{A^{\dagger}A}$. Using the above equation we see that  the $1 \rightarrow 1$ superoperator norm is:
\begin{equation}
	\label{eq12}
	||T||_{1 \rightarrow 1}=\sup_{||A||_{1}=1}||T(A)||_{1},
\end{equation}
\\
where $||A||_{1}=\mathrm{tr}(\sqrt{A^{\dagger}A})$. The $1\rightarrow 1$ superoperator norm also satisfies the following standard properties. For any superoperators $T$ and $V$ we have that,
\begin{align}
	||T+V||_{1\rightarrow 1} \leq ||T||_{1\rightarrow 1}+||V||_{1\rightarrow 1}, && ||TV||_{1\rightarrow 1}\leq ||T||_{1\rightarrow 1}||V||_{1\rightarrow 1}.
\end{align}

Now that we have outlined some notation and defined some basic objects we can state the problem:\\
\\
\textbf{Problem Statement:} Given a continuous one parameter semigroup $\{T_{t}\}$, generated by $\mathcal{L}$, specified by a GKS matrix $A\geq 0 \in \mathcal{M}_{3}(\mathbb{C})$ and a Hamiltonian $H=H^{\dagger} \in \mathcal{M}_{2}(\mathbb{C})$. Find a quantum circuit, using a polynomial number of gates, that approximates $T_{t}=\mathrm{exp}(t\mathcal{L})$ such that the maximum error in the final state, as quantified by the $1\rightarrow1$ superoperator norm is at most $\epsilon$.\\
\\
The strategy to solve this problem consists of three steps. First, we make use of linear combination of semigroups to decompose the channel $T_{t}$ into constituent channels. Then we use unitary conjugation to further decompose the constituent channels into unitary transformations and a channel from some universal semigroup \cite{bacon2001universal}. The second step makes use of SLT product formulas \cite{sweke2014simulation,suzuki1990fractal} to approximate the channel $T_{t}$ using combinations of the constituent channels. Lastly, we convexly decompose the channels from the universal semigroup into extreme channels \cite{ruskai2002analysis} and use quantum forking \cite{park2019parallel} to implement the quantum circuits for the constituent channels.

\section{Decomposition of the arbitrary generator}

The first step is to decompose $\mathcal{L}$ into a combination of generators of simpler semigroups. We do this using linear combination of semigroups and unitary conjugation \cite{bacon2001universal}. Given the GKSL generator for an arbitrary single qubit channel,\\
\begin{equation}
\label{eq13}
 \mathcal{L}(\rho)=-i[H,\rho] + \frac{1}{2}\sum_{i,j=1}^{3}A_{ij}\big( [\sigma_{i},\rho \sigma_{j}] +[\sigma_{i}\rho,\sigma_{j}] \big),
\end{equation}
\\
where $A \geq 0 \in \mathcal{M}_{3}(\mathbb{C})$ and $H=H^{\dagger} \in \mathcal{M}_{2}(\mathbb{C})$. We first need to decompose this generator using linear combination of semigroups. Since $A \geq 0$ we can use the spectral decomposition to write\\
\begin{equation}
\label{eq14}
    A=\sum_{k=1}^{3}\lambda_{k}A_{k}.
\end{equation}
By substituting equation (\ref{eq14}) into equation (\ref{eq13}) and by letting $\mathcal{L}_{0}(\rho)=-i[H,\rho]$ we have,\\
\begin{align}
\label{eq15}
    \mathcal{L}(\rho)&= \mathcal{L}_{0}(\rho) +\frac{1}{2}\sum_{i,j=1}^{3} \bigg( \sum_{k=1}^{3}\lambda_{k}A_{k} \bigg)_{ij} \big(  [\sigma_{i},\rho \sigma_{j}] +[\sigma_{i}\rho,\sigma_{j}] \big),\nonumber \\
    \nonumber\\
    &= \mathcal{L}_{0}(\rho) +\sum_{k=1}^{3}\lambda_{k} \frac{1}{2}\bigg(\sum_{i,j=1}^{3} (A_{k})_{ij}\big(  [\sigma_{i},\rho \sigma_{j}] +[\sigma_{i}\rho,\sigma_{j}] \big) \bigg).
\end{align}
Defining,
\begin{equation}
\label{eq16}
    \mathcal{L}_{k}(\rho)=\frac{1}{2}\sum_{i,j=1}^{3}(A_{k})_{ij}\big(  [\sigma_{i},\rho \sigma_{j}] +[\sigma_{i}\rho,\sigma_{j}] \big),
\end{equation}
we get,\\
\begin{equation}
\label{eq17}
    \mathcal{L}(\rho)=\mathcal{L}_{0}(\rho)+\sum_{k=1}^{3}\lambda_{k}\mathcal{L}_{k}(\rho).
\end{equation}
Letting $\lambda_{0}=1$ we can then write equation (\ref{eq17}) more compactly as\\
\begin{equation}
\label{eq18}
    \mathcal{L}(\rho)=\sum_{k=0}^{3}\lambda_{k}\mathcal{L}_{k}(\rho).
\end{equation}
This allows us to write $T_{t}$ as,\\
\begin{equation}
\label{eq19}
    T_{t}=e^{t\mathcal{L}}=\mathrm{exp}\bigg( t\sum_{k=0}^{3}\lambda_{k}\mathcal{L}_{k} \bigg).
\end{equation}
At this point it is useful for one to split the exponential in equation (\ref{eq19}) into simpler parts which we use in the recombination section to approximate the total channel. This can be achieved via the Lie-Trotter product formula which states that, for any sum of operators $\sum_{k}V_{k}$ acting on some vector space,
\begin{equation}
	\exp(\sum_{k}V_{k})=\lim_{n\rightarrow \infty}\bigg(\prod_{k}\exp(V_{k}/n)\bigg)^{n}.
\end{equation}
If we define $T_{t}^{(k)}:=\mathrm{exp}(t\mathcal{L}_{k})$ and via a straight forward application of the Lie-Trotter product formula \cite{suzuki1990fractal} we have that,\\
\begin{align}
\label{eq20}
    T_{t}&=\lim_{n\rightarrow \infty}\Bigg[ \prod_{k=0}^{3}e^{\big(\frac{t}{n}\lambda_{k}\mathcal{L}_{k} \big)}  \Bigg]^{n} =\lim_{n\rightarrow \infty}\Bigg[ \prod_{k=0}^{3} T_{\big( \frac{t\lambda_{k}}{n}\big)}^{(k)} \Bigg]^{n}.
\end{align}
In the language of \cite{bacon2001universal} we say that $\{T_{t}\}$ can be constructed via a linear combination of semigroups $\{T_{t}^{(k)}\}$. The phrase linear combination of semigroups comes from the fact that we write $\mathcal{L}$ as a linear combination of $\mathcal{L}_{k}$ which generate the semigroups $\{T^{(k)}_{t}\}$. Equation (\ref{eq20}), tells us that one can simulate $T_{t}$ if one can efficiently simulate $T_{t}^{(k)}$ (i.e. the constituent channels) and use $T_{t}^{(k)}$ in some recombination strategy to approximate $T_{t}$.\\
\\
Armed with this realisation that one has to be able to efficiently simulate $T^{(k)}_{t}$ to simulate the channel $T_{t}$. We ask ourselves, can we decompose the channel $T^{(k)}_{t}$ into simpler operations that are easier to efficiently implement on a quantum computer? The answer is yes, Bacon et al. proposed the use of an operation called unitary conjugation to decompose the channel $T^{(k)}_{t}$ into unitary transformations and a channel from some universal semigroup \cite{bacon2001universal}. We now define the operation of unitary conjugation and then prove some results that describe the action of unitary conjugation on a quantum channel. \\
\\
We define unitary conjugation (UC) of a channel $T_{t}$ as the procedure of transforming $T_{t}$ according to $\mathcal{U}^{\dagger}T_{t}\mathcal{U}$ i.e.\\
\begin{equation}
 \label{eq21}
    \mathrm{UC}: T_{t} \mapsto \mathcal{U}^{\dagger}T_{t}\mathcal{U}
\end{equation}
where $\mathcal{U}(X)=UX U^{\dagger}$ for some unitary operator $U$ and any operator $X$. UC preserves all Markovian semigroup properties and it is can be shown that UC effectively applies $T_{t}$ in an alternative basis, however this is not necessary for this tutorial. We recall that $\mathcal{L}_{0}$ ($k=0$), generates Hamiltonian evolution which can be simulated using a single qubit unitary operation. We are interested in the generators of dissipative evolution $\mathcal{L}_{k}$ for $k \in \{1,2,3\}$ and want to use unitary conjugation to further decompose the channels $T^{(k)}_{t}$ for $k \in \{1,2,3\}$. To use UC to decompose $T_{t}^{(k)}$, we need to understand how UC of $T_{t}$ affects the corresponding GKS matrix defining the generator of $T_{t}$. This is made clear by the following theorem from \cite{bacon2001universal}.\\

\begin{theorem}
For a single qubit channel, unitary conjugation of $T_{t}$ by $U \in \mathrm{SU}(2)$ results in conjugation of the GKS matrix $A$ by a corresponding element in $\mathrm{SO}(3)$ which is the adjoint representation of $\mathrm{SU}(2)$.
\end{theorem}

To prove Theorem 1 we need the following two lemmas below, which outline some properties of the map $\mathcal{U}(\rho)$.
\begin{lemma}
Given $\mathcal{U}(\rho)=U\rho U^{\dagger}$ (with $\mathcal{U}^{\dagger}(\rho)=U^{\dagger}\rho U$)  for some unitary operator $U$. $\mathcal{U}(\rho)$ satisfies,\\  
\begin{align}
\mathcal{U}(\mathcal{U}^{\dagger}(\rho))=\mathbb{1}(\rho)=\mathcal{U}^{\dagger}(\mathcal{U}(\rho)).
\end{align}

\end{lemma} 
\begin{proof}
We can prove this directly from the definition of $\mathcal{U}(\rho)$. Recalling that $UU^{\dagger}=U^{\dagger}U=\mathbb{1}$ we get,
\begin{align}
\label{eq22}
    \mathcal{U}(\mathcal{U}^{\dagger}(\rho))&=\mathcal{U}(U^{\dagger}\rho U),\nonumber\\
    &=UU^{\dagger}\rho UU^{\dagger},\nonumber\\\
    &=\rho=\mathbb{1}(\rho).
\end{align}
This implies that, $\mathcal{U}\mathcal{U}^{\dagger}=\mathbb{1}$. Similarly for $\mathcal{U}^{\dagger}\mathcal{U}$ we can show that,
\begin{align}
\label{eq23}
    \mathcal{U}^{\dagger}(\mathcal{U}(\rho))&=\mathcal{U}^{\dagger}(U\rho U^{\dagger}),\nonumber\\
    &=U^{\dagger}U\rho U^{\dagger}U,\nonumber\\
    &=\rho=\mathbb{1}(\rho),
\end{align}
which implies that $\mathcal{U}^{\dagger}\mathcal{U}=\mathbb{1}$. 
\end{proof}
\begin{lemma}
Given $\mathcal{U}(\rho)=U\rho U^{\dagger}$ for some unitary operator $U$, and $p \geq 0$ for $k \in \mathbb{Z}$. We have that,
\begin{equation}
\label{eq24}
    (\mathcal{U}^{\dagger}\mathcal{L}\mathcal{U})^{p}=\mathcal{U}^{\dagger}\mathcal{L}^{p}\mathcal{U}.
\end{equation}
\end{lemma}
\begin{proof}
We prove the formula above using induction. For the base case i.e. $p=0$ we have,\\
\begin{equation}
\label{eq25}
    (\mathcal{U}^{\dagger}\mathcal{L}\mathcal{U})^{0}=\mathbb{1},
\end{equation}
and 
\begin{equation}
\label{eq26}
    \mathcal{U}^{\dagger}\mathcal{L}^{0}\mathcal{U}=\mathcal{U}^{\dagger}\mathbb{1}\mathcal{U}=\mathcal{U}^{\dagger}\mathcal{U}=\mathbb{1},
\end{equation}
where we make use of lemma 1 in the last equality in equation (\ref{eq26}). From equations (\ref{eq25}) and (\ref{eq26}) we see that the formula holds for the base case. Now we assume that the formula is true for $p=m$, that is,
\begin{equation}
\label{eq27}
    (\mathcal{U}^{\dagger}\mathcal{L}\mathcal{U})^{m}=\mathcal{U}^{\dagger}\mathcal{L}^{m}\mathcal{U},
\end{equation}
and we proceed to show that the formula holds for $p=m+1$. We start by writing,
\begin{align}
\label{eq28}
   (\mathcal{U}^{\dagger}\mathcal{L}\mathcal{U})^{m+1}&=(\mathcal{U}^{\dagger}\mathcal{L}\mathcal{U})^{m}(\mathcal{U}^{\dagger}\mathcal{L}\mathcal{U}),\nonumber\\
   &=\mathcal{U}^{\dagger}\mathcal{L}^{m}\mathcal{U}(\mathcal{U}^{\dagger}\mathcal{L}\mathcal{U}), \hspace{5mm} \mathrm{(By\hspace{1mm}eq. \hspace{1mm}(\ref{eq27}))}\nonumber\\
   &=\mathcal{U}^{\dagger}\mathcal{L}^{m}\mathcal{U}\mathcal{U}^{\dagger}\mathcal{L}\mathcal{U},\nonumber\\
   &=\mathcal{U}^{\dagger}\mathcal{L}^{m}\mathcal{L}\mathcal{U}, \hspace{5mm} \mathrm{(By\hspace{1mm}Lemma \hspace{1mm}1)}\nonumber\\
   &=\mathcal{U}^{\dagger}\mathcal{L}^{m+1}\mathcal{U}.
\end{align}
Hence, by induction the formula holds.
\end{proof}
We now have the necessary results to prove Theorem 1.
\begin{proof}[proof of Theorem 1.]
Suppose $\{T_{t}\}$ has a generator $\mathcal{L}$ and a GKS matrix $A$ then $T_{t}=\mathrm{exp}(t\mathcal{L})$. Now for some $U \in \mathrm{SU}(2)$ we want to see what unitary conjugation does to the channel $T_{t}$,
\begin{align}
\label{eq29}
    \mathcal{U}^{\dagger}T_{t}\mathcal{U}&=\mathcal{U}^{\dagger}\mathrm{exp}(t\mathcal{L})\mathcal{U},\nonumber\\
    &=\mathcal{U}^{\dagger}\sum_{p=0}^{\infty}\frac{t^{p}}{p!}\mathcal{L}^{p}\mathcal{U},\nonumber\\
    &=\sum_{k=0}^{\infty}\frac{t^{p}}{p!}\mathcal{U}^{\dagger}\mathcal{L}^{p}\mathcal{U},\nonumber\\
    &=\sum_{k=0}^{\infty}\frac{t^{p}}{p!}(\mathcal{U}^{\dagger}\mathcal{L}\mathcal{U})^{p} ,\hspace{5mm} \mathrm{(By\hspace{1mm}Lemma\hspace{1mm}2)}\nonumber\\
    &=\mathrm{exp}(t\mathcal{U}^{\dagger}\mathcal{L}\mathcal{U}).
\end{align}
So unitary conjugation of $T_{t}$ effectively applies unitary conjugation to the generator $\mathcal{L}$ i.e. $\mathcal{L} \mapsto \mathcal{U}^{\dagger}\mathcal{L}\mathcal{U}$. For the remainder of this proof we make use of the index summation convention where repeated indicies are summed over and the summation range shall always be from 1 to 3 unless otherwise specified. Given the generator in the GKSL form we can always write $\mathcal{L}(\rho)=\mathcal{L}_{H}(\rho)+\mathcal{L}_{D}(\rho)$, where $\mathcal{L}_{H}(\rho)=-i[H,\rho]$ is the Hamiltonian part and $\mathcal{L}_{D}(\rho)=\frac{1}{2}A_{ij}([\sigma_{i},\rho\sigma_{j}]+[\sigma_{i}\rho,\sigma_{j}])$ the dissipative part. We are only interested in how unitary conjugation effects the dissipative part $\mathcal{L}_{D}$ since it contains the GKS matrix $A$. Applying the unitary conjugation map to $\mathcal{L}_{D}$ yields,
\begin{align}
\label{eq31}
    (\mathcal{U}^{\dagger}\mathcal{L}_{D}\mathcal{U})(\rho)&=U^{\dagger}\bigg( \frac{1}{2}A_{ij}\big( [\sigma_{i},U\rho U^{\dagger} \sigma_{j}]+[\sigma_{i}U\rho U^{\dagger},\sigma_{j}] \big) \bigg)U, \nonumber\\
&=\frac{1}{2}A_{ij}\big( U^{\dagger}[\sigma_{i},U\rho U^{\dagger} \sigma_{j}]U+U^{\dagger}[\sigma_{i}U\rho U^{\dagger},\sigma_{j}]U \bigg).
\end{align}

We can simplify the first term in equation (\ref{eq31}) as follows,
\begin{align}
\label{eq32}
    U^{\dagger}[\sigma_{i},U\rho U^{\dagger} \sigma_{j}]U&=U^{\dagger}(\sigma_{i}U\rho U^{\dagger} \sigma_{j}-U\rho U^{\dagger} \sigma_{j}\sigma_{i})U,\nonumber\\
    &=U^{\dagger}\sigma_{i}U\rho U^{\dagger} \sigma_{j}U-\rho U^{\dagger} \sigma_{j}\sigma_{i}U,\nonumber\\
    &=U^{\dagger}\sigma_{i}U\rho U^{\dagger} \sigma_{j}U-\rho U^{\dagger} \sigma_{j}UU^{\dagger}\sigma_{i}U,\nonumber\\
    &=[U^{\dagger}\sigma_{i}U,\rho U^{\dagger}\sigma_{j}U].
\end{align}
In a similar method to the one used in (\ref{eq32}) we show that,
\begin{equation}
\label{eq33}
    U^{\dagger}[\sigma_{i}U\rho U^{\dagger},\sigma_{j}]U=[U^{\dagger}\sigma_{i}U \rho,U^{\dagger}\sigma_{j}U].
\end{equation}
By substituting (\ref{eq32}) and (\ref{eq33}) into (\ref{eq31}) we get,
\begin{align}
\label{eq34}
     (\mathcal{U}^{\dagger}\mathcal{L}_{D}\mathcal{U})(\rho)&=\frac{1}{2}A_{ij}\big([U^{\dagger}\sigma_{i}U,\rho U^{\dagger}\sigma_{j}U]+ [U^{\dagger}\sigma_{i}U \rho,U^{\dagger}\sigma_{j}U]\big).
\end{align}
So unitary conjugation induces a change of basis $\frac{1}{\sqrt{2}}\sigma_{i} \mapsto U^{\dagger}\frac{1}{\sqrt{2}}\sigma_{i}U$, this basis is still Hermitian, orthonormal and traceless. We can expand this basis in terms of the old one as follows,
\begin{equation}
\label{eq35}
    U^{\dagger}\frac{1}{\sqrt{2}}\sigma_{\alpha}U=c_{\alpha \gamma}\frac{1}{\sqrt{2}}\sigma_{\gamma}.
\end{equation}
 By multiply two operators from this basis we can gain more insight into the nature of the matrix $c_{\alpha\gamma}$. First taking the product of two elements in the new basis yields,
\begin{align}
\label{eq36}
    \frac{1}{2}U^{\dagger}\sigma_{\alpha}UU^{\dagger}\sigma_{\beta}U
    &=\frac{1}{2}c_{\alpha \gamma}c_{\beta \nu}^{*}\sigma_{\gamma}\sigma_{\nu}.
\end{align}

Taking the trace of equation (\ref{eq36}) and using the fact that $\mathrm{tr}(\sigma_{\alpha}\sigma_{\beta})=2\delta_{\alpha \beta}$ we have,
\begin{align}
\label{eq37}
    \mathrm{tr}(\frac{1}{2}U^{\dagger}\sigma_{\alpha}UU^{\dagger}\sigma_{\beta}U)&=\frac{1}{2} c_{\alpha \gamma}c_{\beta \nu}^{*}\mathrm{tr}(\sigma_{\gamma}\sigma_{\nu}),\nonumber\\
    &=\frac{1}{2}c_{\alpha \gamma}c_{\beta \nu}^{*}2\delta_{\gamma \nu},\nonumber\\
    &=c_{\alpha \gamma}c_{\beta \gamma}^{*}.
\end{align}
But we also know that the transformed basis $U^{\dagger}\frac{1}{\sqrt{2}}\sigma_{\alpha}U$ is still orthonormal i.e.\\ $\mathrm{tr}(\frac{1}{2}U^{\dagger}\sigma_{\alpha}UU^{\dagger}\sigma_{\beta}U)=\delta_{\alpha\beta}$, leading to the observation that,
\begin{align}
	c_{\alpha \gamma}c_{\beta \gamma}^{*}=\delta_{\alpha\beta}.
\end{align}
In other words $c_{\alpha \gamma}$ is a unitary matrix. Equipped with this information about the matrix $c_{\alpha\beta}$ we can substitute equation (\ref{eq36}) into (\ref{eq35}),
\begin{align}
\label{eq38}
     (\mathcal{U}^{\dagger}\mathcal{L}_{D}\mathcal{U})(\rho)&= \frac{1}{2}A_{\alpha\beta}\big([c_{\alpha \gamma}\sigma_{\gamma},\rho c_{\beta \nu}^{*}\sigma_{\nu}]+ [c_{\alpha \gamma}\sigma_{\gamma}\rho,c_{\beta \nu}^{*}\sigma_{\nu}]\big), \nonumber\\
     &=\frac{1}{2}c_{\alpha \gamma}A_{\alpha\beta}c_{\beta \nu}^{*}\big([\sigma_{\gamma},\rho \sigma_{\nu}]+[\sigma_{\gamma}\rho,\sigma_{\nu}]\big),\nonumber\\
     &=A'_{\gamma\nu}\big([\sigma_{\gamma},\rho \sigma_{\nu}]+[\sigma_{\gamma}\rho,\sigma_{\nu}]\big).
\end{align}
Here $A'$ is the transformed GKS matrix and $A'_{\gamma,\nu}=c_{\alpha \gamma}A_{\alpha \beta}c_{\beta \nu}^{*}$. Defining the matrix $C=c_{\gamma \alpha}$ then we have,
\begin{equation}
\label{eq39}
    A'=C^{T}AC^{*}.
\end{equation}
So the effect of unitary conjugation is to conjugate the GKS matrix by a matrix $C^{T}$. We note that $C^{T}$ is arbitrary but determined by $U \in \mathrm{SU}(2)$ in the following way. Suppose we choose $\{\frac{1}{2}\sigma_{\alpha}\}_{\alpha=1}^{3}$ to be the generators of $\mathrm{SU}(2)$ then the set $\{\frac{1}{2}\sigma_{\alpha}\}_{\alpha=1}^{3}$ forms a basis for the Lie algebra $\mathfrak{su}(2)$ of the Lie group $\mathrm{SU}(2)$. The elements from the basis of $\mathfrak{su}(2)$ also satisfy,
\begin{equation}
\label{eq40}
    \frac{1}{4}[\sigma_{\alpha},\sigma_{\beta}]=i\epsilon_{\alpha \beta \gamma}\frac{1}{2}\sigma_{\gamma},
\end{equation}
where $\epsilon_{\alpha \beta \gamma}$ is the Levi-Cevita symbol and are the structure constants for $\mathfrak{su}(2)$. By setting $U=\mathrm{exp}(ir_{\gamma}\sigma_{\gamma}/2)$ we can expand $U$ in a Taylor series, about zero, to $1^{\mathrm{st}}$ order in an infinitesimal $r_{\gamma}$ and substitute this into the new basis $U^{\dagger}\frac{1}{\sqrt{2}}\sigma_{\alpha}U$,
\begin{align}
\label{eq41}
    U^{\dagger}\frac{1}{\sqrt{2}}\sigma_{\alpha}U&=(\mathbb{1}-i\frac{r_{\gamma}}{2}\sigma_{\gamma})\frac{1}{\sqrt{2}}\sigma_{\alpha}(\mathbb{1}+i\frac{r_{\gamma}}{2}\sigma_{\gamma}),\nonumber\\
    &=\frac{1}{\sqrt{2}}\sigma_{\alpha} +\frac{1}{2\sqrt{2}}ir_{\gamma}\sigma_{\alpha}\sigma_{\gamma}-\frac{1}{2\sqrt{2}}ir_{\gamma}\sigma_{\gamma}\sigma_{\alpha} + \mathcal{O}(r_{\gamma}^{2}),\nonumber\\
    &=\frac{1}{\sqrt{2}}\sigma_{\alpha}-\frac{1}{2\sqrt{2}}ir_{\gamma}[\sigma_{\gamma},\sigma_{\alpha}],\nonumber\\
    &=\frac{1}{\sqrt{2}}\sigma_{\alpha}-ir_{\gamma}(i\epsilon_{\gamma \alpha \beta}\frac{1}{\sqrt{2}}\sigma_{\beta}).
\end{align}
But from equation (\ref{eq35}) $U^{\dagger}\frac{1}{\sqrt{2}}\sigma_{\alpha}U=c_{\alpha \beta}\frac{1}{\sqrt{2}}\sigma_{\beta}$ allowing us to write,
\begin{equation}
\label{eq42}
   c_{\alpha \beta}\frac{1}{\sqrt{2}}\sigma_{\beta}=\frac{1}{\sqrt{2}}\sigma_{\alpha}-ir_{\gamma}(i\epsilon_{\gamma \alpha \beta}\frac{1}{\sqrt{2}}\sigma_{\beta}).
\end{equation}
Multiplying both sides of equation (\ref{eq42}) from the left by $\frac{1}{\sqrt{2}}\sigma_{\beta '}$ we have,
\begin{equation}
\label{eq43}
     \frac{1}{2}c_{\alpha \beta}\sigma_{\beta '}\sigma_{\beta}=\frac{1}{2}\sigma_{\beta '}\sigma_{\alpha}-\frac{1}{2}ir_{\gamma}(i\epsilon_{\gamma \alpha \beta}\sigma_{\beta '}\sigma_{\beta}).
\end{equation}
Finally taking the trace of equation (\ref{eq43}) and using the fact that $\mathrm{tr}(\sigma_{\beta}\sigma_{\alpha})=2\delta_{\alpha \beta}$ we obtain,
\begin{align}
\label{eq44}
 c_{\alpha \beta}\delta_{\beta ' \beta}&=\delta_{\beta ' \alpha}-ir_{\gamma}(i\epsilon_{\gamma \alpha \beta}\delta_{\beta ' \beta})
    \implies c_{\alpha \beta '}=\delta_{\alpha \beta ' }-ir_{\gamma}(i\epsilon_{\gamma \alpha \beta '})=(C)_{\alpha \beta '}.
\end{align}
It is clear from (\ref{eq44}) that $C^{T}$ is,
\begin{equation}
\label{eq45}
    (C^{T})_{\alpha \beta}=(C)_{\beta \alpha}=\delta_{\beta \alpha}-ir_{\gamma}(i\epsilon_{\gamma \beta \alpha}),
\end{equation}
but the Levi-Cevita symbol is totally anti-symmetric allowing one to write,
\begin{align}
\label{eq46}
    (C^{T})_{\alpha \beta}=\delta_{\alpha \beta}+ir_{\gamma}(i\epsilon_{\gamma \alpha \beta}). 
\end{align}
We observe that $C^{T}$ is in a Lie group generated by $(G_{\gamma})_{\alpha \beta}=i\epsilon_{\gamma \alpha \beta}$ which implies that \\$C^{T}=\mathrm{exp}(ir_{\gamma}G_{\gamma} )$. Since we know how each $G_{\gamma}$ is defined by picking $\gamma,\alpha,\beta \in \{1,2,3\}$ we can explicitly calculate the matrix elements for each $G_{\gamma}$ as,
\begin{align}
\label{eq48}
    G_{1}=i\begin{pmatrix}
          0 & 0 & 0\\
          0 & 0 & 1\\
          0 & -1 & 0\\
        \end{pmatrix}, \hspace{5mm} 
        G_{2}=i\begin{pmatrix}
         0 & 0 & -1\\
         0 & 0 & 0\\
         1 & 0 & 0\\
        \end{pmatrix},\hspace{5mm}
         G_{3}=i&\begin{pmatrix}
         0 & 1 & 0\\
         -1 & 0 & 0\\
         0 & 0 & 0\\
        \end{pmatrix}.
\end{align}
We recognize the above matrices as the generators of $\mathrm{SO}(3)$, hence $C^{T}\in \mathrm{SO}(3)$ which is the adjoint representation of $\mathrm{SU}(2)$. As a brief mathematical aside: it is elementary fact of group theory that the adjoint representation of the Lie algebra of $\mathrm{SU}(N)$ is generated by the structure constants of the Lie algebra of $\mathrm{SU}(N)$ (For more information on the adjoint representation and its use in this proof one could refer to Appendix A. as well as the following books \cite{hall2013lie,gilmore2008lie,schuller2015lectures}). Therefore, unitary conjugation of the channel $T_{t}$ by $U \in \mathrm{SU}(2)$ leads to conjugation of the GKS matrix $A$ by $C^{T} \in \mathrm{SO}(3)$ i.e. $\mathcal{U}^{\dagger}T_{t}\mathcal{U}$ leads to a new GKS matrix $A'=C^{T}AC$.
\end{proof}
\noindent So far we have seen how linear combination of semigroups allows us to write channels from the semigroup $\{T_{t}\}$ as a limit of a  product of constituent channels from the semigroup $\{T^{(k)}_{t}\}$. It was also shown how unitary conjugation affects the channel $T_{t}$. Theorem 2 below will show how using linear combination of semigroups and unitary conjugation one can simulate a channel $T_{t}$, it is a modification of theorem 3 in \cite{bacon2001universal}. This modification was made so that we can account for the differing notation in the quantum information community as well as the open quantum systems community. As we know researchers working in the field of open quantum systems define the ground state of a single qubit as the column vector $(1,0)^{T}$, while the quantum information community defines the excited state with this column vector. The range of the parameter $\theta$ in Theorem 2. below accounts for this differing definitions for ground and excited states.
\begin{theorem}
To simulate, using linear combination and unitary conjugation, an arbitrary Markovian semigroup generated by $\mathcal{L} \in \mathcal{B}(\mathcal{B}(\mathcal{H}_{s}))$ with $\mathcal{H}_{s} \cong \mathbb{C}^{2}$, it is necessary and sufficient to be able to simulate all Markovian semigroups whose generators are specified by the GKS matrix $A(\theta)=\vec{a}(\theta)\vec{a}(\theta)^{\dagger}$ where $\vec{a}(\theta)=(\cos(\theta),-i\sin (\theta),0)^{T}$ and $\theta \in [-\frac{\pi}{4},\frac{\pi}{4}]$.
\end{theorem}
\begin{proof}
(Sufficiency) First, without any loss of generality we assume $H = 0$. Let $A \geq 0 \in \mathcal{M}_{3}(\mathbb{C})$, be the GKS matrix specifying the generator of the Markovian semigroup we wish to simulate. We fix the basis $\{F_{i}\}=\frac{1}{\sqrt{2}}\{ \sigma_{i}\}_{i=1}^{3}$ without any loss of generality. Since $A \geq 0$ we use the spectral decomposition to express $A$ as in equation (\ref{eq14}). Since we can always write each projector $A_{k}$ as an outer product i.e. $A_{k}=\vec{a}_{k} \vec{a}_{k}^{\dagger}$, we have,
\begin{equation}
\label{eq49}
    A=\sum_{k=1}^{3}\lambda_{k}\vec{a}_{k} \vec{a}_{k}^{\dagger},
\end{equation}
where $\lambda_{k}\geq 0$ and $\vec{a}_{k}^{\dagger}\cdot \vec{a}_{k}=1$ for $k=1,2,3$. By linear combination of semigroups it is sufficient to simulate all GKS matrices $\vec{a}_{k}\vec{a}_{k}^{\dagger}$ with $|\vec{a}_{k}|=1$, for simplicity from this point onwards we drop the subscript $k$ and just write $\vec{a}\vec{a}^{\dagger}$ and $|\vec{a}|=1$. Any vector $\vec{a}$ can be split into real and imaginary parts, $\vec{a}^{R}$ and $\vec{a}^{I}$ respectively,
\begin{equation}
\label{eq50}
    \vec{a}=\vec{a}^{R}+i\vec{a}^{I}.
\end{equation}
Since $\vec{a}$ only appears in outer products multiplying $\vec{a}$ by an overall phase $e^{i\psi}$ leaves $\vec{a}\vec{a}^{\dagger}$ invariant i.e. if $\vec{a}'=e^{i\psi}\vec{a}$ then,
\begin{equation}
\label{eq51}
    \vec{a}'\vec{a}'^{\dagger}=e^{i\psi}\vec{a}e^{-i\psi}\vec{a}^{\dagger}=\vec{a}\vec{a}^{\dagger}.
\end{equation}
Multiplying by an overall phase changes $\vec{a}$ as follows,
\begin{align}
\label{eq52}
    \vec{a}'&=e^{i\psi}\vec{a},\nonumber\\
    &=e^{i\psi}\vec{a}^{R}+ie^{i\psi}\vec{a}^{I},\nonumber\\
    &=\vec{a}'^{R}+i\vec{a}'^{I},
\end{align}
where $\vec{a}'^{R}=\vec{a}^{R}\cos(\psi)-\vec{a}^{I}\sin(\psi)$ and $\vec{a}'^{I}=\vec{a}^{R}\sin(\psi)+\vec{a}^{I}\cos(\psi)$. So if we wish to simulate $\vec{a}\vec{a}^{\dagger}$ we could simulate $\vec{a}'\vec{a}'^{\dagger}$ for any value of $\psi$. Defining the following two parameters,
\begin{equation}
\label{eq53}
    k_{1}=|\vec{a}^{R}|^{2}-|\vec{a}^{I}|^{2} \hspace{3mm} \mathrm{and} \hspace{3mm} k_{2}=2\vec{a}^{R}\cdot\vec{a}^{I},
\end{equation}
and multiplying $\vec{a}$ by an overall phase changes $k_{1}$ and $k_{2}$ as follows,
\begin{align}
\label{eq54}
    k_{1}'&=|\vec{a}'^{R}|^{2}-|\vec{a}'^{I}|^{2},\nonumber\\
    &=|\vec{a}^{R}\cos(\psi)-\vec{a}^{I}\sin(\psi)|^{2}-|\vec{a}^{R}\sin(\psi)+\vec{a}^{I}\cos(\psi)|^{2},\nonumber\\
    &=|\vec{a}^{R}|^{2}\cos(2\psi) -|\vec{a}^{I}|\cos(2\psi)-\sin(2\psi)(2\vec{a}^{R}\cdot\vec{a}^{I}),\nonumber\\
    &=\cos(2\psi)k_{1}-\sin(2\psi)k_{2},
\end{align}
and
\begin{align}
\label{eq55}
    k_{2}'&=2\vec{a}'^{R}\cdot\vec{a}'^{I},\nonumber\\
    &=2(\vec{a}^{R}\cos(\psi)-\vec{a}^{I}\sin(\psi))\cdot(\vec{a}^{R}\sin(\psi)+\vec{a}^{I}\cos(\psi)),\nonumber\\
    &=\sin(2\psi)(|\vec{a}^{R}|^{2}-|\vec{a}^{I}|^{2}) +\cos(2\psi)(2\vec{a}^{R}\cdot\vec{a}^{I}),\nonumber\\
    &=\sin(2\psi)k_{1}+\cos(2\psi)k_{2}.
\end{align}

We can express the transformation of $k_{1}$ and $k_{2}$ as a matrix equation,
\begin{equation}
\label{eq56}
    \begin{pmatrix}
    k_{1}' \\
    k_{2}'
    \end{pmatrix}=\begin{pmatrix}
    \cos(2\psi) & -\sin(2\psi) \\
    \sin(2\psi) & \cos(2\psi)
    \end{pmatrix}\begin{pmatrix}
    k_{1}\\
    k_{2}
    \end{pmatrix}.
\end{equation}
Since we can choose $\psi$ arbitrarily we can make the choice,
\begin{equation}
\label{eq57}
    \tan(2\psi)=-\frac{k_{2}}{k_{1}},
\end{equation}
such that $k_{2}'=0$, in which case $\vec{a}'^{R}\cdot\vec{a}'^{I}=0$. In addition we can choose, $k_{1}'=k_{1}/\cos(2\psi) \geq 0$, such that $|\vec{a}'^{R}|\geq |\vec{a}'^{I}|$. Hence due to the phase freedom in $\vec{a}$ we can assume, without loss of generality, that $\vec{a}^{R}\cdot\vec{a}^{I}=0$ and $|\vec{a}^{R}|\geq |\vec{a}^{I}|$. Since $\vec{a}^{R}$ and $\vec{a}^{I}$ are orthogonal and $|\vec{a}|=1$ we can parameterize $\vec{a}$ as,
\begin{equation}
\label{eq58}
    \vec{a}=\vec{a}^{R}+i\vec{a}^{I}=\cos(\theta)\hat{a}^{R}+i\sin(\theta)\hat{a}^{I},
\end{equation}
where $|\hat{a}^{R}|=|\hat{a}^{I}|=1$ are unit vectors and $\theta \in [-\frac{\pi}{4},\frac{\pi}{4}]$. Now performing a unitary transformation on $\vec{a}\vec{a}^{\dagger}$ by $G \in \mathrm{SO}(3)$ we get,
\begin{equation}
\label{59}
    \vec{a}\vec{a}^{\dagger} \xrightarrow[]{G} G\vec{a}\vec{a}^{\dagger}G^{T} = (G\vec{a})(G\vec{a})^{\dagger}.
\end{equation}
Since $G \in \mathrm{SO}(3)$ is real it does not mix the real and imaginary parts of $\vec{a}$, so $G$ simultaneously rotates the vectors $\vec{a}^{R}$ and $\vec{a}^{I}$. Hence $G$ can be chosen to align the unit vectors $\hat{a}^{R}$ and $\hat{a}^{I}$ to the $+x$ and $-y$ axis respectively, that is, 
\begin{equation}
\label{eq60}
    G\hat{a}^{R}=\hat{x} \hspace{3mm} \mathrm{and} \hspace{3mm} G\hat{a}^{I}=-\hat{y}.
\end{equation}
So we can write,
\begin{equation}
\label{eq61}
    \vec{a}(\theta)=G\vec{a}=(\cos(\theta),-i\sin(\theta),0)^{T},
\end{equation}
where $\theta \in [-\frac{\pi}{4},\frac{\pi}{4}]$ and,
\begin{equation}
\label{eq62}
    A(\theta)=\vec{a}(\theta)\vec{a}(\theta)^{\dagger}=\begin{pmatrix}
    \cos^{2}(\theta) & i\cos(\theta)\sin(\theta) & 0\\
    -i\cos(\theta)\sin(\theta) & \sin^{2}(\theta) & 0\\
    0 & 0 & 0\\
    \end{pmatrix}.
\end{equation}
Hence we see that using the methods of linear combination of semigroups and unitary conjugation to simulate a Markovian semigroup, whose generator is specified by the GKS matrix $A$, it suffices to be able to simulate all Markovian semigroups whose generator is specified by the GKS matrix $A(\theta)$.\\
\\
(Necessity) We now show that using linear combination and unitary conjugation it is not possible to simulate the Markovian semigroup whose generator is specified by the GKS matrix $A(\theta)$, through simulation of some other combination and/or transformation of Markovian semigroups whose generators are specified by GKS matrices of the form of $A(\theta')$.\\
\\
Firstly we note that $A(\theta)=\vec{a}(\theta)\vec{a}(\theta)^{\dagger}$ is a rank 1 matrix and a projector onto the eigenspace of a single eigenvector of $A$. Since rank 1 matrices are extreme in the cone of positive semidefinite matrices \cite{agler1988positive}, $A(\theta)$ cannot be written as a convex combination of positive semidefinite matrices. Hence no such $A(\theta)$ can be simulated using linear combination of semigroups that are specified by other such matrices $A(\theta')$.\\
\\
We now note that multiplying the vector $\vec{a}$ by an overall phase $e^{i\psi}$ and the rotation $G \in \mathrm{SO}(3)$ commute. This tells us that to show that we cannot use unitary conjugation to simulate the Markovian semigroup whose generator is specified by $A(\theta)$, it suffices to show that, for an $\psi \in [0,2\pi]$ and $G \in \mathrm{SO}(3)$ if\\
\begin{equation}
\label{eq63}
    e^{i\psi}G\vec{a}(\theta)=\vec{a}(\theta'),
\end{equation}
\\
where $\theta,\theta' \in [-\frac{\pi}{4},\frac{\pi}{4}]$, then $\theta=\theta'$. To see this let $k_{2}'=k_{2}=0$ and $k_{1}',k_{1} \geq 0$ in equation (\ref{eq56}), this leads to the following two equations,\\
\begin{equation}
\label{eq64}
    k_{1}'=\cos(2\psi)k_{1} \hspace{3mm} \mathrm{and} \hspace{3mm} \sin(2\psi)k_{1}=0.
\end{equation}
\\
Since by assumption $k_{1}',k_{1} \geq 0$ we see that these equations are only satisfied when $\psi=0$, hence the phase transformation $e^{i\psi}$ is trivial and leaves $\theta$ unchanged. Now we know that $\vec{a}(\theta)=\cos(\theta)\hat{x}-i\sin(\theta)\hat{y}$ so,\\
\begin{equation}
\label{eq65}
    G\vec{a}(\theta)=\cos(\theta)G\hat{x}+i\sin(\theta)G(-\hat{y})=\cos(\theta)\hat{x}-i\sin(\theta)\hat{y}.
\end{equation}
\\
From the above equation we see that $G$ does not change $\theta$, also since $G$ was chosen to align the unit vectors $\hat{a}^{R}$ and $\hat{a}^{I}$ to the unit vectors $+\hat{x}$ and $-\hat{y}$ respectively. We see that $G\vec{a}(\theta)=\vec{a}(\theta)$. Hence we see that,\\
\begin{equation}
\label{eq66}
    e^{i\psi}G\vec{a}(\theta)=a(\theta),
\end{equation}\\
and hence $\theta=\theta'$. Therefore it is not possible to use linear combination of semigroups and unitary conjugation to simulate the Markovian semigroup, whose generator is specified $A(\theta)$, with Markovian semigroups whose generators are specified by matrices of the same form $A(\theta')$. 
\end{proof}
\noindent Now by making use of Theorem. 2 we can decompose the constituent channels $T_{t}^{(k)}$. We start with the GKS matrix $A_{k}=\vec{a}_{k}\vec{a}_{k}^{\dagger}$ then by choosing some $C_{k} \in \mathrm{SO}(3)$, and by performing a unitary transformation on $\vec{a}_{k} \vec{a}_{k}^{\dagger}$ with $C_{K}$ we get:
\begin{equation}
\label{eq67}
    C_{k}A_{k}C_{k}^{T}= C_{k}\vec{a}_{k} \vec{a}_{k}^{\dagger}C_{k}^{T}=(C_{k}\vec{a}_{k})(C_{k}\vec{a}_{k})^{\dagger}=\vec{a}(\theta_{k})\vec{a}(\theta_{k})^{\dagger}=A(\theta_{k}).
\end{equation}
The above equation implies that,
\begin{equation}
\label{eq68}
     A_{k}=C_{k}^{T}A(\theta_{k})C_{k}.
\end{equation}
Now we define the channel $T_{t}^{(\theta_{k})}:=\exp(t \mathcal{L}_{\theta_{k}})$, where,
\begin{equation}
\label{eq69}
	    \mathcal{L}_{\theta_{k}}(\rho)=\frac{1}{2}\sum_{i,j=1}^{3}(A(\theta_{k}))_{ij}\big( [\sigma_{i},\rho \sigma_{j}]+[\sigma_{i}\rho,\sigma_{j}]\big),
\end{equation}
and the GKS matrix $A(\theta_{k})$ is given in equation (\ref{eq62}). Using the definitions above we see that unitary conjugation of the channel $T_{t}^{(\theta_{k})}$ by $U_{k} \in \mathrm{SU}(2)$ conjugates the GKS matrix $A(\theta_{k})$ by some $C_{k} \in \mathrm{SO}(3)$ i.e. equation (\ref{eq68}). Hence unitary conjugation of the channel $T_{t}^{(\theta_{k})}$ yields the channel $T_{t}^{(k)}$ i.e. 
\begin{equation}
\label{eq70}
    T_{t}^{(k)}(\rho)=(\mathcal{U}_{k}^{\dagger}T_{t}^{(\theta_{k})}\mathcal{U}_{k})(\rho)=U_{k}^{\dagger}[T_{t}^{(\theta_{k})}(U_{k}\rho U_{k}^{\dagger})]U_{k}.
\end{equation}
This tells us that if we wish to simulate the channels from the semigroup $\{T_{t}^{(k)}\}$ we need to be able to efficiently simulate channels from the universal semigroup $\{ T_{t}^{(\theta_{k})}\}$ and apply a unitary transformation $U_{k} \in \mathrm{SU(2)}$. In the next section we shall discuss a recombination strategy to use the constituent channels $T^{(k)}_{t}$ in second order SLT product formulas to approximate the channel $T_{t}$.

\section{Recombination Strategy for Approximating \texorpdfstring{$T_{t}$}{}}

In this section, using methods developed for simulating Hamiltonian dynamics \cite{berry2007efficient,papageorgiou2012efficiency}, 
% Add references to all papers I have on Hamiltonian simulation using SLT product formulas there are quite a few!!!
we wish to show that we can use the second order Suzuki-Lie-Trotter (SLT) product formulas \cite{suzuki1990fractal,suzuki1991general} 
% add Suzuki papers here!!
to simulate the channel $T_{t}$ up to an arbitrary accuracy $\epsilon$. In particular we wish to use a finite product of elements from the semigroup $\{T^{(k)}_{t}\}$ with $T^{(k)}_{t}:=e^{t\mathcal{L}_{k}}$, to approximate the channel $T_{t}$ and we also want to place a bound on the number of implementations of $T^{(k)}_{t}$ required in this product. \\
\\
\noindent Given the generator of the channel $T_{t}$ in the form as in equation (\ref{eq18}) we can absorb the scalars $\lambda_{k}$ into the superoperators $\mathcal{L}_{k}$ by defining, $\Lk{k}:=\lambda_{k}\mathcal{L}_{k}$ and writing the generator as,
\begin{align}
    \label{eq146}
    \mathcal{L}=\sum^{3}_{k=0}\Lk{k}.
\end{align}
Using the $1\rightarrow 1$ superoperator norm we define the quantity $\Lambda:= \max_{k \in \{0,1,2,3\}}\snorm{\Lk{k}}$, this allows us to bound the $1\rightarrow 1$ norm of the generator $\snorm{\mathcal{L}}$ in the following way,
\begin{align}
    \label{eq147}
\snorm{\mathcal{L}}&=\snorm{\left(\sum_{k=0}^{3}\Lk{k}\right) }\leq \sum_{k=0}^{3}\snorm{\Lk{k}} \leq 4\Lambda.
\end{align}
We consider the total evolution $T_{t}=e^{t\mathcal{L}}$ and by defining the parameter $\tau:=t/N$ we can write one $N^{\mathrm{th}}$ of the total evolution as,
\begin{align}
\label{eq148}
    T_{\tau}=e^{\tau\mathcal{L}}=\exp\bigg(\tau\sum_{k=0}^{3}\Lk{k}\bigg).
\end{align}

Now we define the second order SLT product formula $S_{2}(\alpha)$ for some $\alpha \geq 0$\cite{suzuki1990fractal,suzuki1991general} as,
\begin{align}
    \label{eq149}
    S_{2}(\alpha)=S_{2}(\Lk{0},\Lk{1},\Lk{2},\Lk{3},\alpha)=\prod_{k=0}^{3}e^{\frac{\alpha}{2}\Lk{k}}\prod_{k'=3}^{0}e^{\frac{\alpha}{2}\Lk{k'}}.
\end{align}

Now to calculate the bound on the error in the approximation of the total channel $T_{t}$ we need to find the bound on the difference $\snorm{T_{\tau}-S_{2}(\tau)}$ and then consider $N$ applications of the SLT integrator $S_{2}(\tau)$. The following theorem shall encapsulate the bound on the error in our approximation as well as the number of implementations $N$ required to approximate the total channel.\\
\begin{theorem}
    Given a quantum channel $T_{t}=e^{t\mathcal{L}}$ with a generator $\mathcal{L}=\sum_{k=0}^{3}\Lk{k}$, and for $0 \leq \epsilon \leq 1$ then there exists some $N$ such that,
    \begin{align}
        \label{eq150}
        \snorm{\exp(\sum_{k=0}^{3}\Lk{k})-S_{2}(\tau)^{N}}\leq \epsilon,
    \end{align}
    with $\epsilon \geq \frac{(4t \Lambda)^{3}}{3N^{2}}$ and $\tau=\frac{t}{N}$.
\end{theorem}

To prove the following theorem we require some results which are stated in the following lemmas below.\\
\begin{lemma}
    For $\tau=\frac{t}{N}>0$ we have that,
    \begin{align}
    \label{eq151}
        \snorm{\exp(\sum_{k=0}^{3}\Lk{k})-S_{2}(\tau)}\leq \frac{(4t\Lambda)^{3}}{3N^{3}}.
    \end{align}
    \end{lemma}
\begin{proof}
    We know that $S_{2}(\tau)$ will approximate $T_{\tau}$ up to order $\tau^{2}$ i.e. $\snorm{T_{\tau}-S_{2}(\tau)}\in \mathcal{O}(\tau^{3})$, so we have that in the difference $T_{\tau}-S_{2}(\tau)$ the only remaining terms in the Taylor expansion are of order $\tau^{3}$ so we have
    \begin{align}
        \label{eq152}
        \exp(\tau\sum_{k=0}^{3}\Lk{k})-S_{2}(\tau)=\sum_{l=3}^{\infty}R_{l}(\tau)-W_{l}(\tau),
    \end{align}
    where $R_{l}(\tau)$ are the remaining terms in the Taylor expansion of $\exp(\tau\sum_{k=0}^{3}\Lk{k})$ and $W_{l}(\tau)$ are the remaining terms in the Taylor expansion of $S_{2}(\tau)$. We can then write $R_{l}(\tau)$ as,
    \begin{align}
        \label{eq153}
        R_{l}(\tau)=\frac{\tau^{l}}{l!}\bigg(\sum_{k=0}^{3}\Lk{k}\bigg)^{l}.
    \end{align}
    We can then bound this quantity by using the fact that $\snorm{\mathcal{L}}\leq 4\Lambda$, so that we are able to write,
    \begin{align}
        \label{eq154}
        \snorm{R_{l}(\tau)}&=\snorm{\frac{\tau^{l}}{l!}\mathcal{L}^{l}},\nonumber\\
        &= \frac{\tau^{l}}{l!}\snorm{\mathcal{L}^{l}},\nonumber\\
        &\leq \frac{\tau^{l}}{l!}\snorm{\mathcal{L}}^{l},\nonumber\\
        &\leq \frac{\tau^{l}}{l!}(4\Lambda)^{l}.
    \end{align}
    To write the form of the term $W_{l}(\tau)$ we need to consider the Taylor expansion of $S_{2}(\tau)$, which is given in terms of the product of Taylor expansions of each exponential in equation (\ref{eq149}),
\begin{align}
    \label{eq155}
    S_{2}(\tau)&=\sum_{j_{1},...,j_{8}=0}^{\infty}\frac{(\tau/2)^{j_{1}+...+j_{8}}}{j_{1}!j_{2}!...j_{8}!}\Lk{0}^{j_{1}}\Lk{1}^{j_{2}}\Lk{2}^{j_{3}}\Lk{3}^{j_{4}}\Lk{3}^{j_{5}}\Lk{2}^{j_{6}}\Lk{1}^{j_{7}}\Lk{0}^{j_{8}},\\
    \label{eq156}
    &=\sum_{l=0}^{\infty}\hspace{1mm}\sum\limits_{\substack{j_{1},...,j_{8}=0 \\ \sum_{m} j_{m}=l}}^{l}\frac{(\tau/2)^{j_{1}+...+j_{8}}}{j_{1}!j_{2}!...j_{8}!}\Lk{0}^{j_{1}}\Lk{1}^{j_{2}}\Lk{2}^{j_{3}}\Lk{3}^{j_{4}}\Lk{3}^{j_{5}}\Lk{2}^{j_{6}}\Lk{1}^{j_{7}}\Lk{0}^{j_{8}},
\end{align}
where we have converted the eight infinite sums in (\ref{eq156}) to finite sums by restricting $j_{1},...,j_{8}$ such that they sum to $l$. This allows us to write the term $W_{l}(\tau)$ as,
\begin{align}
    \label{eq157}
    W_{l}(\tau)=\sum\limits_{\substack{j_{1},...,j_{8}=0 \\ \sum_{m} j_{m}=l}}^{l}\frac{(\tau/2)^{j_{1}+...+j_{8}}}{j_{1}!j_{2}!...j_{8}!}\Lk{0}^{j_{1}}\Lk{1}^{j_{2}}\Lk{2}^{j_{3}}\Lk{3}^{j_{4}}\Lk{3}^{j_{5}}\Lk{2}^{j_{6}}\Lk{1}^{j_{7}}\Lk{0}^{j_{8}}.
\end{align}
To bound the term $W_{l}(\tau)$ we use the fact that $\snorm{\cdot}$ is sub-multiplicative and sub-additive, so that we can write,
\begin{align}
\label{eq158}
    \snorm{W_{l}(\tau)}&=\snorm{\sum\limits_{\substack{j_{1},...,j_{8}=0 \\ \sum_{m} j_{m}=l}}^{l}\frac{(\tau/2)^{j_{1}+...+j_{8}}}{j_{1}!j_{2}!...j_{8}!}\Lk{0}^{j_{1}}\Lk{1}^{j_{2}}\Lk{2}^{j_{3}}\Lk{3}^{j_{4}}\Lk{3}^{j_{5}}\Lk{2}^{j_{6}}\Lk{1}^{j_{7}}\Lk{0}^{j_{8}}},\nonumber\\
    &\leq \sum\limits_{\substack{j_{1},...,j_{8}=0 \\ \sum_{m} j_{m}=l}}^{l}\frac{(\tau/2)^{j_{1}+...+j_{8}}}{j_{1}!j_{2}!...j_{8}!}\snorm{\Lk{0}^{j_{1}}\Lk{1}^{j_{2}}\Lk{2}^{j_{3}}\Lk{3}^{j_{4}}\Lk{3}^{j_{5}}\Lk{2}^{j_{6}}\Lk{1}^{j_{7}}\Lk{0}^{j_{8}}},\nonumber\\
    &\leq \sum\limits_{\substack{j_{1},...,j_{8}=0 \\ \sum_{m} j_{m}=l}}^{l}\frac{(\tau/2)^{j_{1}+...+j_{8}}}{j_{1}!j_{2}!...j_{8}!}\snorm{\Lk{0}^{j_{1}}}\snorm{\Lk{1}^{j_{2}}}...\snorm{\Lk{0}^{j_{8}}},\nonumber\\
    &\leq \sum\limits_{\substack{j_{1},...,j_{8}=0 \\ \sum_{m} j_{m}=l}}^{l}\frac{(\tau/2)^{j_{1}+...+j_{8}}}{j_{1}!j_{2}!...j_{8}!}\snorm{\Lk{0}}^{j_{1}}\snorm{\Lk{1}}^{j_{2}}...\snorm{\Lk{0}}^{j_{8}},\nonumber\\
    &\leq \sum\limits_{\substack{j_{1},...,j_{8}=0 \\ \sum_{m} j_{m}=l}}^{l}\frac{(\tau/2)^{j_{1}+...+j_{8}}}{j_{1}!j_{2}!...j_{8}!}\Lambda^{j_{1}+j_{2}+...+j_{8}}.
\end{align}
To calculate the sum in equation (\ref{eq158}) we must consider the taylor expansion of the following expression,
\begin{align}
\label{eq159}
    \exp(\frac{x}{2})^{8}&=\exp(\frac{x}{2})\exp(\frac{x}{2})...\exp(\frac{x}{2}),\nonumber\\
    &=\bigg(\sum_{j_{1}=0}^{\infty}\frac{(x/2)^{j_{1}}}{j_{1}!}\bigg)\bigg(\sum_{j_{2}=0}^{\infty}\frac{(x/2)^{j_{2}}}{j_{2}!}\bigg)...\bigg(\sum_{j_{8}=0}^{\infty}\frac{(x/2)^{j_{8}}}{j_{8}!}\bigg),\nonumber\\
    &=\sum_{j_{1},j_{2},...,j_{8}=0}^{\infty}\frac{1}{j_{1}!j_{2}!...j_{8}!}\bigg(\frac{x}{2}\bigg)^{j_{1}+j_{2}+...+j_{8}},\nonumber\\
    &=\sum_{p=0}^{\infty}\hspace{1mm}\sum\limits_{\substack{j_{1},j_{2},...,j_{8}=0\\\sum_{m}j_{m}=p}}^{p}\frac{1}{j_{1}!j_{2}!...j_{8}!}\bigg(\frac{x}{2}\bigg)^{j_{1}+j_{2}+...+j_{8}}.
\end{align}
But we know that,
\begin{align}
\label{eq160}
    \exp(\frac{x}{2})^{8}=\exp(4x)=\sum_{p=0}^{\infty}\frac{4^{p}x^{p}}{p!}.
\end{align}
Now by equating equations (\ref{eq159}) and (\ref{eq160}) we see that,
\begin{align}
\label{eq161}
    \sum\limits_{\substack{j_{1},j_{2},...,j_{8}=0\\\sum_{m}j_{m}=p}}^{p}\frac{1}{j_{1}!j_{2}!...j_{8}!}\bigg(\frac{x}{2}\bigg)^{j_{1}+j_{2}+...+j_{8}}=\frac{4^{p}x^{p}}{p!}.
\end{align}
By using $x=\tau\Lambda$ in equation (\ref{eq161})we can write the sum in equation (\ref{eq158}) as,
\begin{align}
\label{eq162}
\sum\limits_{\substack{j_{1},...,j_{8}=0 \\ \sum_{m} j_{m}=l}}^{l}\frac{(\tau/2)^{j_{1}+...+j_{8}}}{j_{1}!j_{2}!...j_{8}!}\Lambda^{j_{1}+j_{2}+...+j_{8}}=\frac{4^{l}\tau^{l}\Lambda^{l}}{2^{l}l!},
\end{align}
this implies that the bound on $\snorm{W_{l}(\tau)}$ is,
\begin{align}
    \label{eq163}
    \snorm{W_{l}(\tau)}\leq \frac{4^{l}\tau^{l}\Lambda^{l}}{l!}.
\end{align}
Now, using sub-additive property of the 1$\rightarrow$1 superoperator norm, we can write
\begin{align}
    \label{eq164}
    \snorm{\exp(\tau\sum_{k=0}^{3}\Lk{k})-S_{2}(\tau)}&=\snorm{\sum_{l=3}^{\infty}R_{l}(\tau)-W_{l}(\tau)},\nonumber\\
    &\leq\sum_{l=3}^{\infty}\snorm{R_{l}(\tau)}+\snorm{W_{l}(\tau)},\nonumber\\
    &\leq2\sum_{l=3}\frac{(4\tau\Lambda)^{l}}{l!}.
\end{align}
Now by using the Lemma F.2 from the supplementary information of \cite{childs2018toward} we have that for $y>0$,
\begin{align}
    \label{eq165}
    \sum_{n=k}^{\infty}\frac{y^{n}}{n!}\leq \frac{y^{k}}{k!}\exp(y).
\end{align}
Using this upper bound on the remainder terms in the Taylor expansion of the exponential function we can write,
\begin{align}
    \label{eq166}
    \snorm{\exp(\tau\sum_{k=0}^{3}\Lk{k})-S_{2}(\tau)}&\leq 2 \frac{(4\tau\Lambda)^{3}}{3!}\exp(4\tau\Lambda)=\frac{(4t\Lambda)^{3}}{3N^{3}}\exp(\frac{4t\Lambda}{N}).
\end{align}
For large enough $N$ we can approximate $\exp(\frac{4t\Lambda}{N})\approx 1$ which gives the desired result,
\begin{align}
    \label{eq167}
     \snorm{\exp(\tau\sum_{k=0}^{3}\Lk{k})-S_{2}(\tau)}&\leq\frac{(4t\Lambda)^{3}}{3N^{3}},
\end{align}
completing the proof.
\end{proof}
\begin{lemma}
    For quantum channels $T$ and $V$ and for some $N\geq0 \in \mathbb{Z}$,
    \begin{align}
        \label{eq168}
        \snorm{T^{N}-V^{N}}\leq N\snorm{T-V}.
    \end{align}
\end{lemma}

\begin{proof}
    We prove this by induction. For $N=0,1$ the equality in equation (\ref{eq168}) holds and to show this would be trivial, so for the base case in our inductive proof we choose $N=2$,
    \begin{align}
        \label{eq169}
        \snorm{T^{2}-V^{2}}&=\snorm{T^{2}-TV+TV-V^{2}},\nonumber\\
        &=\snorm{T(T-V)+(T-V)V},\nonumber\\
        &\leq \snorm{T(T-V)}+\snorm{(T-V)V},\nonumber\\
        &=\leq \snorm{T}\snorm{T-V}+\snorm{T-V}\snorm{V}.
    \end{align}
    By recalling that for any quantum channel $T$ by definition $\snorm{T}=1$, this allows us to write,
    \begin{align}
        \label{eq170}
        \snorm{T^{2}-V^{2}}&\leq \snorm{T-V}+\snorm{T-V}=2\snorm{T-V}.
    \end{align}
    Hence we have verified that the inequality in equation (\ref{eq168}) holds for $N=2$, we now assume that it holds for $N=m$ and show that it is true for $N=m+1$.
    \begin{align}
        \label{eq171}
        \snorm{T^{m+1}-V^{m+1}}&=\snorm{T^{m+1}-TV^{m}+TV^{m}-V^{m+1}},\nonumber\\
        &=\snorm{T(T^{m}-V^{m})+(T-V)V^{m}},\nonumber\\
        &\leq \snorm{T}\snorm{T^{m}-V^{m}}+\snorm{T-V}\snorm{V^{m}},\nonumber\\
        &\leq m\snorm{T-V}+\snorm{T-V}\snorm{V}^{m},\\
        &\leq (m+1)\snorm{T-V}.
    \end{align}
    Therefore by induction the inequality in (\ref{eq168}) holds true for all integers $N\geq0$.
\end{proof}
We are now able to write the proof for theorem 5. using the above two lemmas.
\begin{proof}{of Theorem 5.}
Given that $\tau=t/N$ we can write $\exp(t\mathcal{L})=\exp(\tau\mathcal{L})^{N}$ which allows us to write,
\begin{align}
    \label{eq172}
    \snorm{\exp(t\mathcal{L})-(S_{2}(\tau))^{N}}&=\snorm{\exp(\tau\mathcal{L})^{N}-(S_{2}(\tau))^{N}},\nonumber\\
    &\leq N\snorm{\exp(\tau\mathcal{L})-S_{2}(\tau)},\nonumber\\
    &\leq N \frac{(4t\Lambda)^{3}}{3N^{3}}=\frac{(4t\Lambda)^{3}}{3N^{2}},
\end{align}
where in the second line of (\ref{eq172}) we used lemma 6. and in the third line we use lemma 5. By choosing $\epsilon \geq \frac{(4t\Lambda)^{3}}{3N^{2}}$ we get that $\snorm{\exp(t\mathcal{L})-(S_{2}(\tau))^{N}}\leq \epsilon$, which completes the proof.
\end{proof}

Finally, using the fact that $\epsilon \geq \frac{(4t\Lambda)^{3}}{3N^{2}}$ we can obtain a bound on the number of implementations of $T^{(k)}_{t}$ required to simulate $T_{t}$ up to a precision $\epsilon$. We first see that,
\begin{align}
	N\geq\frac{(4t\Lambda)^{3/2}}{(3\epsilon)^{1/2}},
\end{align} 
this tells us that we have at least $\frac{(4t\Lambda)^{3/2}}{(3\epsilon)^{1/2}}$  implementations of $S_{2}(\tau)$, and in each $S_{2}(\tau)$ we have seven constituent channels $T^{(k)}_{t}$ which implies that the number of channels $T^{(k)}_{t}$ needed to simulate the channel $T_{t}$, which will be denoted by $N_{exp}$ is,
\begin{align}
	N_{exp}\geq 7N =\frac{7(4t\Lambda)^{3/2}}{(3\epsilon)^{1/2}}.
\end{align}
This places a bound on the number of implementations of $T^{(k)}_{t}$ needed to simulate $T_{t}$ and this bound will be used in the next sections to place a bound on the number of single qubit and CNOT gates needed to simulate the channel $T_{t}$.

\section{Simulation of the Constituent Channels}
We have seen from equation (\ref{eq70}), that to simulate the channels $T_{t}^{(k)}$ we need to be able to simulate channels from the semigroup $\{ T_{t}^{(\theta_{k})}\}$ as well as apply some unitary transformation. We note that since the unitary transformation comes from $\mathrm{SU}(2)$ it can be easily implemented with a single qubit unitary gate. So we are interested in finding a quantum circuit that can simulate channels from the semigroup $\{ T_{t}^{(\theta_{k})}\}$.\\
\\
We do this using the techniques from \cite{ruskai2002analysis} to decompose the channel $T_{t}^{(\theta_{k})}$ into a convex sum of quasi-extreme channels, which are channels that are generalised extreme points of the space of quantum channels and have only 2 Kraus operators. We can then use the Stinespring representation \cite{stinespring1955positive} to find quantum circuits that correspond to these quasi-extreme channels. Lastly we use the method of quantum forking \cite{park2019parallel} to implement the convex mixture of these quasi-extreme channels. This will give us a quantum circuit that simulates the channel $T_{t}^{(\theta_{k})}$ for some $\theta_{k} \in [-\frac{\pi}{4},\frac{\pi}{4}]$.\\
\subsection{Calculation of the Choi matrix of the channel \texorpdfstring{$T_{t}^{(\theta_{k})}$}{}}
We now consider the channel $T_{t}^{(\theta_{k})}=\exp(t\mathcal{L}_{\theta_{k}})$, whose generator is given in equation (\ref{eq69}). We now make use of the matrix representation of the generator to find the matrix representation of the channel. We shall denote the matrix representation of a channel or its generator by bold face. The generator $\mathcal{L}_{\theta_{k}}$ has the following matrix representation,
\begin{equation}
	\label{eq71}
	(\mathbf{L}_{\theta_{k}})_{a,b}=\mathrm{tr}[G_{a}^{\dagger}\mathcal{L}_{\theta_{k}}(G_{b})],
\end{equation}
where the operators $\{G_{a}\}$ is a basis for the space $\mathcal{M}_{2}(\mathbb{C})$, and is given as $\{G_{1}=\frac{1}{\sqrt{2}}\mathbb{1}, G_{2}=\frac{1}{\sqrt{2}}\sigma_{1},\\  G_{3}=\frac{1}{\sqrt{2}}\sigma_{2},  G_{4}=\frac{1}{\sqrt{2}}\sigma_{3}\}$. Now we note that the choice of this basis leads to a matrix representation of the channel that is called the affine map representation, refer to appendix A for more details on this representation of the quantum channel. Now we can use equation (\ref{eq69}) and (\ref{eq71}) to find the matrix representation of the generator $\mathcal{L}_{\theta_{k}}$ as,
\begin{equation}
	\label{eq72}
	\mathbf{L}_{\theta_{k}}=\begin{pmatrix}
		0 & 0 & 0 & 0\\
		0 & -2\sin^{2}(\theta_{k}) & 0  & 0\\
		0 & 0 & -2\cos^{2}(\theta_{k}) & 0\\
		-4\cos(\theta_{k})\sin(\theta_{k}) & 0 & 0 & -2\\
	\end{pmatrix}.
\end{equation}
Now to find the matrix representation of the channel $T_{t}^{(\theta_{k})}$ we make use of the following,
\begin{equation}
	\label{eq73}
	\mathbf{T}_{t}^{(\theta_{k})}=\exp(t\mathbf{L}_{\theta_{k}}).
\end{equation}
This leads to the following matrix representation of $T_{t}^{(\theta_{k})}$,
\begin{equation}
	\label{eq74}
	\mathbf{T}_{t}^{(\theta_{k})}=\begin{pmatrix}
		1 & 0 & 0 & 0\\
		0 & \lambda_{1} & 0 & 0\\
		0 & 0 & \lambda_{2} & 0\\
		m_{3} & 0 & 0 & \lambda_{3}\\
	\end{pmatrix},
\end{equation}
where
\begin{align}
	\label{eq75}
	\lambda_{1}&=\exp(-2t\sin^{2}(\theta_{k})),\nonumber\\
	\lambda_{2}&=\exp(-2t\cos^{2}(\theta_{k})),\nonumber\\
	\lambda_{3}&=\exp(-2t),\nonumber\\
	m_{3}&=\sin(2\theta_{k})(\lambda_{3}-1).
\end{align}
Now to obtain the desired convex decomposition of the channel $T_{t}^{(\theta_{k})}$, we need to find the Choi matrix \cite{jamiolkowski1972linear} of this channel i.e $\tau_{(\theta_{k})}=(T_{t}^{(\theta_{k})} \otimes \mathbb{1})|\Omega \rangle \langle \Omega|$, where $|\Omega\rangle=\frac{1}{\sqrt{2}}(|00\rangle +|11\rangle)$. We do this using the following relation \cite{wolf2012quantum},
\begin{equation}
	\label{eq76}
	\tau_{(\theta_{k})}=\frac{1}{4}\sum_{i,j=0}^{3}(\mathbf{T}_{t}^{(\theta_{k})})_{i+1,j+1}(\sigma_{i} \otimes \sigma_{j}^{T}),
\end{equation}
where $\sigma_{0}=\mathbb{1}$. Using equation (\ref{eq76}) we find the Choi matrix to be,
\begin{equation}
	\label{eq77}
	\tau_{(\theta_{k})}=\frac{1}{4}\begin{pmatrix}
		a^{2} & 0 & 0 & \lambda_{1}+\lambda_{2}\\
		0 & b^{2} & \lambda_{1}-\lambda_{2} & 0\\
		0 & \lambda_{1}-\lambda_{2} & c^{2} & 0\\
		\lambda_{1}+\lambda_{2} & 0 & 0 & d^{2}\\
	\end{pmatrix},
\end{equation}
with
\begin{align}
	\label{eq78}
	a&=(1 + \lambda_{3} + m_{3})^{\frac{1}{2}},\nonumber\\
	b&=(1 - \lambda_{3} + m_{3})^{\frac{1}{2}},\nonumber\\
	c&=(1 - \lambda_{3} - m_{3})^{\frac{1}{2}},\nonumber\\
	d&=(1 + \lambda_{3} - m_{3})^{\frac{1}{2}}.
\end{align}
\subsection{Convex decomposition of \texorpdfstring{$T_{t}^{(\theta_{k})}$}{} into quasi-extreme channels}
At this point it is useful to adopt the notation used in \cite{ruskai2002analysis} and define the Choi matrix in terms of the matrix in terms of the matrix $\beta(T_{t}^{(\theta_{k})})=2\tau_{(\theta_{k})}$, this implies that $\tau_{(\theta_{k})}=\frac{1}{2}\beta(T_{t}^{(\theta_{k})})$. \\
The $\beta(\cdot)$ matrix defined in terms of the Choi matrix plays a significant role as this definition implies that $\beta$ defines an affine isomorphism between space of all CPTP maps on the space of states of a single qubit, denoted by $\mathcal{Q}$, and the image $\beta(\mathcal{Q}) \subset \mathcal{M}_{4}(\mathbb{C})$. In particular, there is a one-to-one correspondence between extreme points of $\mathcal{Q}$ and those of the image $\beta(\mathcal{Q})$. This will correspondence will allow us to determine the extreme points of $\mathcal{Q}$, which will aid in decomposing $T^{(\theta_{k})}_{t}$ into quasi-extreme channels.\\  
From this we can also find the Choi matrix of the dual channel $\tilde{T}_{t}^{(\theta_{k})}$ using the following relation,
\begin{equation}
	\label{eq79}
	\beta(\tilde{T}_{t}^{(\theta_{k})})=(U_{23}^{\dagger}\beta(T_{t}^{(\theta_{k})})U_{23})^{*},
\end{equation}
where $U_{23}$ is the permutation matrix that swaps the second and third rows of a $4 \times 4$ matrix and is defined as,
\begin{equation}
	\label{eq80}
	U_{23}=U_{23}^{\dagger}=\begin{pmatrix}
		1 & 0 & 0 & 0\\
		0 & 0 & 1 & 0\\
		0 & 1 & 0 & 0\\
		0 & 0 & 0 & 1\\
	\end{pmatrix}.
\end{equation}
Using equation (\ref{eq79}) we find the matrix $\beta(\tilde{T}_{t}^{(\theta_{k})})$,
\begin{equation}
	\label{eq81}
	\beta(\tilde{T}_{t}^{(\theta_{k})})=\frac{1}{2}\begin{pmatrix}
		a^{2} & 0 & 0 & \lambda_{1}+\lambda_{2}\\
		0 & c^{2} & \lambda_{1}-\lambda_{2} & 0\\
		0 & \lambda_{1}-\lambda_{2} & b^{2} & 0\\
		\lambda_{1}+\lambda_{2} & 0 & 0 & d^{2}\\
	\end{pmatrix}.
\end{equation}
Now the Choi matrix corresponding to the dual channel, $\tilde{T}_{t}^{(\theta_{k})}$, is $\tilde{\tau}_{(\theta_{k})}=\frac{1}{2}\beta(\tilde{T}_{t}^{(\theta_{}k)})$. We shall now make use of the following two lemmas and theorem from \cite{ruskai2002analysis}, to obtain the convex decomposition of $\tau_{(\theta_{k})}$, we state them below in a modified form to suit our notation.\\
\begin{lemma}
	Any contraction in $\mathcal{M}_{2}(\mathbb{C})$ can be written as the convex combination of two unitary matrices.
\end{lemma}
\begin{proof}
	Let $R\in \mathcal{M}_{2}(\mathbb{C})$ be a contraction then $||R|| \leq 1$, and its singular value decomposition can be written in the form
	\begin{equation}
	\label{eq82}
		R=V\begin{pmatrix}
			\cos(\theta_{1}) & 0\\
			0 & \cos(\theta_{2})\\
		\end{pmatrix}W^{\dagger},
	\end{equation}
where $V,W^{\dagger} \in \mathrm{U}(2)$. Now by using the fact that $\cos(\theta)=\frac{1}{2}(e^{i\theta}+e^{-i\theta})$, we have that,
\begin{align}
\label{eq83}
	R&=V\frac{1}{2}\begin{pmatrix}
		e^{i\theta_{1}}+e^{-i\theta_{1}} & 0\\
		0 & e^{i\theta_{2}}+e^{-i\theta_{2}}\\
	\end{pmatrix}W^{\dagger},\nonumber\\
	&=\frac{1}{2}V\begin{pmatrix}
		e^{i\theta_{1}} & 0\\
		0 & e^{i\theta_{2}}\\
	\end{pmatrix}W^{\dagger}+\frac{1}{2}V\begin{pmatrix}
		e^{-i\theta_{1}} & 0\\
		0 & e^{-i\theta_{2}}\\
	\end{pmatrix}W^{\dagger}.
\end{align}
From the above it is clear that a contraction $R$ can be written as a convex combination of two unitary matrices.
\end{proof}
\begin{lemma}
	Given a matrix
	\begin{equation}
	\label{eq84}
		J=\begin{pmatrix}
			A & C\\
			C^{\dagger} & B\\
		\end{pmatrix},
	\end{equation}
	(1) J is positive semidefinite if and only if $A\geq 0$, $B \geq 0$ and $C=\sqrt{A}R\sqrt{B}$ for some contraction $R$.\\
	(2) Moreover the set of positive semidefinite matrices with fixed $A$ and $B$ is a convex set whose extreme points satisfy $C=\sqrt{A}U\sqrt{B}$, for some unitary matrix $U$.
	\end{lemma}
\begin{proof}
	The proof of (1) follows from \cite{horn1994topics}, suppose that $A,B \geq 0$ then by theorem (7.7.7) in \cite{horn1990matrix} the block matrix $J$ is positive semidefinite if and only if
	\begin{align}
	\label{eq85}
		1 \geq \varrho(C^{\dagger}A^{-1}CB^{-1})&=\varrho(C^{\dagger}A^{-1}CB^{-\frac{1}{2}}B^{-\frac{1}{2}}),\\
		&=\varrho(B^{-\frac{1}{2}}C^{\dagger}A^{-1}CB^{-\frac{1}{2}}),\hspace{5mm} \text{since }\varrho(AB)=\varrho(BA)\hspace{2mm} \forall A,B\in \mathcal{M}_{n}(\mathbb{C})\nonumber\\
		&=\varrho(B^{-\frac{1}{2}}C^{\dagger}A^{-\frac{1}{2}}A^{-\frac{1}{2}}CB^{-\frac{1}{2}}),\\
		&=\varrho((A^{-\frac{1}{2}}CB^{-\frac{1}{2}})^{\dagger}(A^{-\frac{1}{2}}CB^{-\frac{1}{2}})),\nonumber\\
		&=\varsigma(A^{-\frac{1}{2}}CB^{-\frac{1}{2}})^{2}, \hspace{5mm}\text{since } \varrho(A^{\dagger}A)=\varsigma(A)^{2},
	\end{align}
where $\varrho(\cdot)$ is the spectral radius and $\varsigma(\cdot)$ the spectral norm. Setting $R=A^{-\frac{1}{2}}CB^{-\frac{1}{2}}$ (since $\varsigma(R) \leq 1 \implies R$ is a contraction) we have $C=\sqrt{A}R\sqrt{B}$ as desired.\\
The proof of (2) follows from a well know result that the extreme points of the set of contractions in $\mathcal{M}_{2}(\mathbb{C})$ are unitary (see \cite{horn1994topics,pendersen1979}). Now if we define the affine mapping $f : R \mapsto \sqrt{A}R\sqrt{B}$ for a fixed $A,B \geq 0$, where $\mathrm{dom}(f)$ is the set of positive semidefinite matrices in $\mathcal{M}_{2}(\mathbb{C})$ and $\mathrm{Im}(f) \subset \mathcal{M}_{2}(\mathbb{C})$. Now the set of contractions in $\mathcal{M}_{2}(\mathbb{C})$ is a compact and convex set, and extreme points of the $\mathrm{Im}(f)$ are images of extreme points \cite{ruskai2002analysis}. If $R \in \mathrm{Dom}(f)$ is an extreme point then it is unitary, and setting $U=R$ we have that $C=f(R)=f(U)=\sqrt{A}U\sqrt{B}$ \cite{ruskai2002analysis,horn1994topics}.
\end{proof}
The following theorem, from \cite{ruskai2002analysis}, is stated without proof as the proof is lengthy and can be found in \cite{ruskai2002analysis}.
\begin{theorem}
	A quantum channel $T_{t}$ is a generalised extreme point of the set of all quantum channels of the same dimension if and only if $\beta(\tilde{T}_{t}) $ is of the form,
	\begin{equation}
	\label{eq86}
		\beta(\tilde{T}_{t})=\begin{pmatrix}
			A & \sqrt{A}U\sqrt{B}\\
			\sqrt{B}U^{\dagger}\sqrt{A} & B
		\end{pmatrix},
		\end{equation}
		for some unitary matrix $U$.
\end{theorem}
We can now write find the convex decomposition of $\tau_{(\theta_{k})}$ . We start by writing the matrix $\beta(\tilde{T}_{t}^{(\theta_{k})})$ in block form as follows:
\begin{equation}
\label{eq87}
	\beta(\tilde{T}_{t}^{(\theta_{k})})=\begin{pmatrix}
		A & C\\
		C^{\dagger} & B\\
	\end{pmatrix},
\end{equation}
where 
\begin{align}
	\label{eq88}
	A=\frac{1}{2}\begin{pmatrix}
		a^{2} & 0\\
		0 & c^{2}
	\end{pmatrix}, && B=\frac{1}{2}\begin{pmatrix}
		b^{2} & 0\\
		0 & d^{2}\\
	\end{pmatrix}, && C=\frac{1}{2}\begin{pmatrix}
		0 & \lambda_{1}+\lambda_{2}\\
		\lambda_{1}-\lambda_{2} & 0\\
	\end{pmatrix}.
\end{align}
Now since $\beta(\tilde{T}_{t}^{(\theta_{k})}) \geq 0$, we can make use of Lemma 4. so that we can write $C=\sqrt{A}R\sqrt{B}$, for some contraction $R \in \mathcal{M}_{2}(\mathbb{C})$, which leads to,
 \begin{equation}
 \label{eq89}
	\beta(\tilde{T}_{t}^{(\theta_{k})})=\begin{pmatrix}
		A & \sqrt{A}R\sqrt{B}\\
		\sqrt{B}R^{\dagger}\sqrt{A} & B\\
	\end{pmatrix}.
\end{equation}
Now we can compute the contraction $R$ from the fact that $C=\sqrt{A}R\sqrt{B}$,
\begin{align}
	\label{eq90}
	R&=A^{-\frac{1}{2}}CB^{-\frac{1}{2}},\nonumber\\
	&=\sqrt{2}\begin{pmatrix}
		\frac{1}{a} & 0\\
		0 & \frac{1}{c}\\
	\end{pmatrix} \frac{1}{2}\begin{pmatrix}
		0 & \lambda_{1}+\lambda_{2}\\
		\lambda_{1}-\lambda_{2} & 0\\
	\end{pmatrix}\sqrt{2}\begin{pmatrix}
		\frac{1}{b} & 0\\
		0 & \frac{1}{d}\\
	\end{pmatrix} ,\nonumber\\
	&=\begin{pmatrix}
		0 & \frac{\lambda_{1}+\lambda_{2}}{ad}\\
		\frac{\lambda_{1}-\lambda_{2}}{bc} & 0\\
	\end{pmatrix}.
\end{align}
Now by making use of Lemma 3. we know that any contraction in $\mathcal{M}_{2}(\mathbb{C})$ can be written as a convex sum of two unitary matrices, so if we let $V=\mathbb{1}$ and $W^{\dagger}=\sigma_{1}$ in equation (\ref{eq82}) we have:
\begin{align}
	\label{eq91}
	R=\begin{pmatrix}
		0 & \cos(\phi_{1})\\
		\cos(\phi_{2}) & 0\\
	\end{pmatrix}=\begin{pmatrix}
		0 & \frac{\lambda_{1}+\lambda_{2}}{ad}\\
		\frac{\lambda_{1}-\lambda_{2}}{bc} & 0\\
	\end{pmatrix},
\end{align}
this gives the angles $\phi_{1}$ and $\phi_{2}$ as,
\begin{align}
	\label{eq92}
	\phi_{1}=\arccos(\frac{\lambda_{1}+\lambda_{2}}{ad}), && \phi_{2}=\arccos(\frac{\lambda_{1}-\lambda_{2}}{bc}).
\end{align}
By making use of equation (\ref{eq83}) we can write $R$ as,
\begin{align}
	\label{eq93}
	R=\frac{1}{2}U_{1}+\frac{1}{2}U_{2},
\end{align}
where
\begin{align}
	\label{eq94}
	U_{1}=\begin{pmatrix}
		0 & e^{i\phi_{1}}\\
		e^{i\phi_{2}} & 0\\
	\end{pmatrix}, && 	U_{2}=\begin{pmatrix}
		0 & e^{-i\phi_{1}}\\
		e^{-i\phi_{2}} & 0\\
	\end{pmatrix}.
\end{align}
We can now write,
\begin{align}
	\label{eq95}
	\beta(\hat{T}_{t}^{(\theta_{k})})&=\begin{pmatrix}
		A & \sqrt{A}(\frac{1}{2}U_{1}+\frac{1}{2}U_{2})\sqrt{B}\\
		\sqrt{B}(\frac{1}{2}U_{1}^{\dagger}+\frac{1}{2}U_{2}^{\dagger})\sqrt{A} & B\\
	\end{pmatrix},\nonumber\\
	\nonumber\\
	&=\frac{1}{2}\begin{pmatrix}
		A & \sqrt{A}U_{1}\sqrt{B}\\
		\sqrt{B}U_{1}^{\dagger}\sqrt{A} & B\\
	\end{pmatrix}+\frac{1}{2}\begin{pmatrix}
		A & \sqrt{A}U_{2}\sqrt{B}\\
		\sqrt{B}U_{2}^{\dagger}\sqrt{A} & B\\
	\end{pmatrix}.
\end{align}
If we define,
\begin{equation}
	\label{eq96}
	\beta(\tilde{T}_{(t,j)}^{(\theta_{k})})=\begin{pmatrix}
		A & \sqrt{A}U_{j}\sqrt{B}\\
		\sqrt{B}U_{j}^{\dagger}\sqrt{A} & B\\
	\end{pmatrix},
\end{equation}
for $j=1,2$, then we can write:
\begin{equation}
	\label{eq97}
	\beta(\tilde{T}_{t}^{(\theta_{k})})=\frac{1}{2}\beta(\tilde{T}_{(t,1)}^{(\theta_{k})})+\frac{1}{2}\beta(\tilde{T}_{(t,2)}^{(\theta_{k})}).
\end{equation}
Now by Theorem 4. the matrix $\beta(\tilde{T}_{(t,j)}^{(\theta_{k})})$ is a generalised extreme point so the channel $T_{(t,j)}^{(\theta_{k})}$ is a generalised extreme point in the space of channels \cite{ruskai2002analysis}, also known as a quasi extreme channel. We can now begin to find the convex decomposition of the Choi matrix $\tau_{(\theta_{k})}$ in terms of the Choi matrices of these quasi extreme channels, as follows. By making use of equation (\ref{eq79}) as well as the fact that $\tau_{(\theta_{k})}=\frac{1}{2}\beta(T_{t}^{(\theta_{k})})$ we can find the desired convex decomposition of the Choi matrix $\tau_{(\theta_{k})}$ in the following way. From equation (\ref{eq79}) and (\ref{eq97}) we have that,
\begin{align}
	\label{eq98}
	\beta(T_{t}^{(\theta_{k})})&=U_{23}\beta(\tilde{T}_{t}^{(\theta_{k})})^{*}U_{23}^{\dagger},\nonumber\\
	&=U_{23}(\frac{1}{2}\beta(\tilde{T}_{(t,1)}^{(\theta_{k})})^{*}+\frac{1}{2}\beta(\tilde{T}_{(t,2)}^{(\theta_{k})})^{*})U_{23}^{\dagger},\nonumber\\
	&=\frac{1}{2}U_{23}\beta(\tilde{T}_{(t,1)}^{(\theta_{k})})^{*}U_{23}^{\dagger}+\frac{1}{2}U_{23}\beta(\tilde{T}_{(t,2)}^{(\theta_{k})})^{*}U_{23}^{\dagger},\nonumber\\
	&=\frac{1}{2}\beta(T_{(t,1)}^{(\theta_{k})})+\frac{1}{2}\beta(T_{(t,2)}^{(\theta_{k})}),
\end{align}
where $\beta(T_{(t,j)}^{(\theta_{k})})=U_{23}\beta(\tilde{T}_{(t,j)}^{(\theta_{k})})^{*}U_{23}^{\dagger}$, for $j=1,2$. Now by defining, $\tau_{(\theta_{k},j)}=\frac{1}{2}\beta(T_{(t,j)}^{(\theta_{k})})$, we can write the convex decomposition of $\tau_{(\theta_{k})}$ as,
\begin{equation}
	\label{eq99}
	\tau_{(\theta_{k})}=\frac{1}{2}\tau_{(\theta_{k},1)}+\frac{1}{2}\tau_{(\theta_{k},2)}.
\end{equation}
Using equations (\ref{eq79}), (\ref{eq88}), (\ref{eq94}) and (\ref{eq96}) we can find the Choi matrices $\tau_{(\theta_{k},j)}$ for $j=1,2$ to be,
\begin{align}
	\label{eq100}
	\tau_{(\theta_{k},1)}&=\frac{1}{4}\begin{pmatrix}
		a^{2} & 0 & 0 & ade^{-i\phi_{1}}\\
		0 & b^{2} & bce^{i\phi_{2}} & 0 \\
		0 &  bce^{-i\phi_{2}} & c^{2} & 0\\
		ade^{i\phi_{1}} & 0 & 0 & d^{2}\\
	\end{pmatrix},\nonumber\\
	\nonumber\\
	\tau_{(\theta_{k},2)}&=\frac{1}{4}\begin{pmatrix}
		a^{2} & 0 & 0 & ade^{i\phi_{1}}\\
		0 & b^{2} & bce^{-i\phi_{2}} & 0 \\
		0 &  bce^{i\phi_{2}} & c^{2} & 0\\
		ade^{-i\phi_{1}} & 0 & 0 & d^{2}\\
	\end{pmatrix}.
\end{align}
Now by the Choi-Jamiolkowski isomorphism \cite{choi1975completely} and equation (\ref{eq99}), we have the desired convex decomposition of the channel $T_{t}^{(\theta_{k})}$ as a convex sum of two quasi extreme channels i.e.
\begin{equation}
	\label{eq101}
	T_{t}^{(\theta_{k})}(\rho)=\frac{1}{2}T_{(t,1)}^{(\theta_{k})}(\rho)+\frac{1}{2}T_{(t,2)}^{(\theta_{k})}(\rho).
\end{equation}
Using the Choi matrices in equation (\ref{eq100}), we can find the Kraus maps corresponding to the channels $T_{(t,1)}^{(\theta_{k})}$ and $T_{(t,2)}^{(\theta_{k})}$, we do this as follows. We make use of the following formula from \cite{wolf2012quantum}, for any Choi matrix $\tau=(T \otimes \mathbb{1})|\Omega\rangle\langle\Omega|$ we have $\tau \geq 0$ which implies,
\begin{equation}
	\label{eq102}
	\tau=\sum_{j}|\psi_{j}\rangle\langle \psi_{j}|=\sum_{j}(K_{j} \otimes \mathbb{1})|\Omega\rangle\langle \Omega|(K_{j} \otimes \mathbb{1})^{\dagger},
\end{equation}
where each $\ket{\psi_{j}}$ is the product of $j$-th eigenvector of $\tau$ with its corresponding eigenvalue. 
Using the formula above and equation (\ref{eq100}) we can obtain the Kraus maps for the channels $T_{(t,1)}^{(\theta_{k})}$ and $T_{(t,2)}^{(\theta_{k})}$. For the channel $T_{(t,1)}^{(\theta_{k})}$ we have,
\begin{align}
\label{eq103}
	T_{(t,1)}^{(\theta_{k})}(\rho)=\sum_{j=1}^{2}K_{j}^{(1)}\rho K_{j}^{(1)\dagger},
\end{align}
where
\begin{align}
	\label{eq104}
	K_{1}^{(1)}=\frac{1}{\sqrt{2}}\begin{pmatrix}
		ae^{-i\phi_{1}} & 0 \\
		0 & d\\
	\end{pmatrix}, && K_{2}^{(1)}=\frac{1}{\sqrt{2}}\begin{pmatrix}
		0 & be^{i\phi_{2}}\\
		c & 0\\
	\end{pmatrix}.
\end{align}
For the channel $T_{(t,2)}^{(\theta_{k})}$ we have,
\begin{align}
\label{eq105}
	T_{(t,2)}^{(\theta_{k})}(\rho)=\sum_{j=1}^{2}K_{j}^{(2)}\rho K_{j}^{(2)\dagger},
\end{align}
where
\begin{align}
	\label{eq106}
	K_{1}^{(2)}=\frac{1}{\sqrt{2}}\begin{pmatrix}
		ae^{i\phi_{1}} & 0 \\
		0 & d\\
	\end{pmatrix}, && K_{2}^{(2)}=\frac{1}{\sqrt{2}}\begin{pmatrix}
		0 & be^{-i\phi_{2}}\\
		c & 0\\
	\end{pmatrix}.
\end{align}
\subsection{Construction of circuits for the simulation of \texorpdfstring{$T_{t}^{(k)}$}{}}
Now that we have obtained a decomposition of $T_{t}^{(\theta_{k})}$ as a convex sum of quasi extreme channels we now begin to find a quantum circuit that can implement this convex sum of two quasi extreme channels so that we can simulate the channel $T_{t}^{(\theta_{k})}$ on a quantum computer. To do this we first need to find unitary operators that correspond to the quasi extreme channels, for this we make use of the Stinespring Representation of the channel \cite{stinespring1955positive}. For more information on the Stinespring representation for quantum channels refer to Appendix C.

\begin{figure*}
\centering
	\includegraphics[scale=0.30]{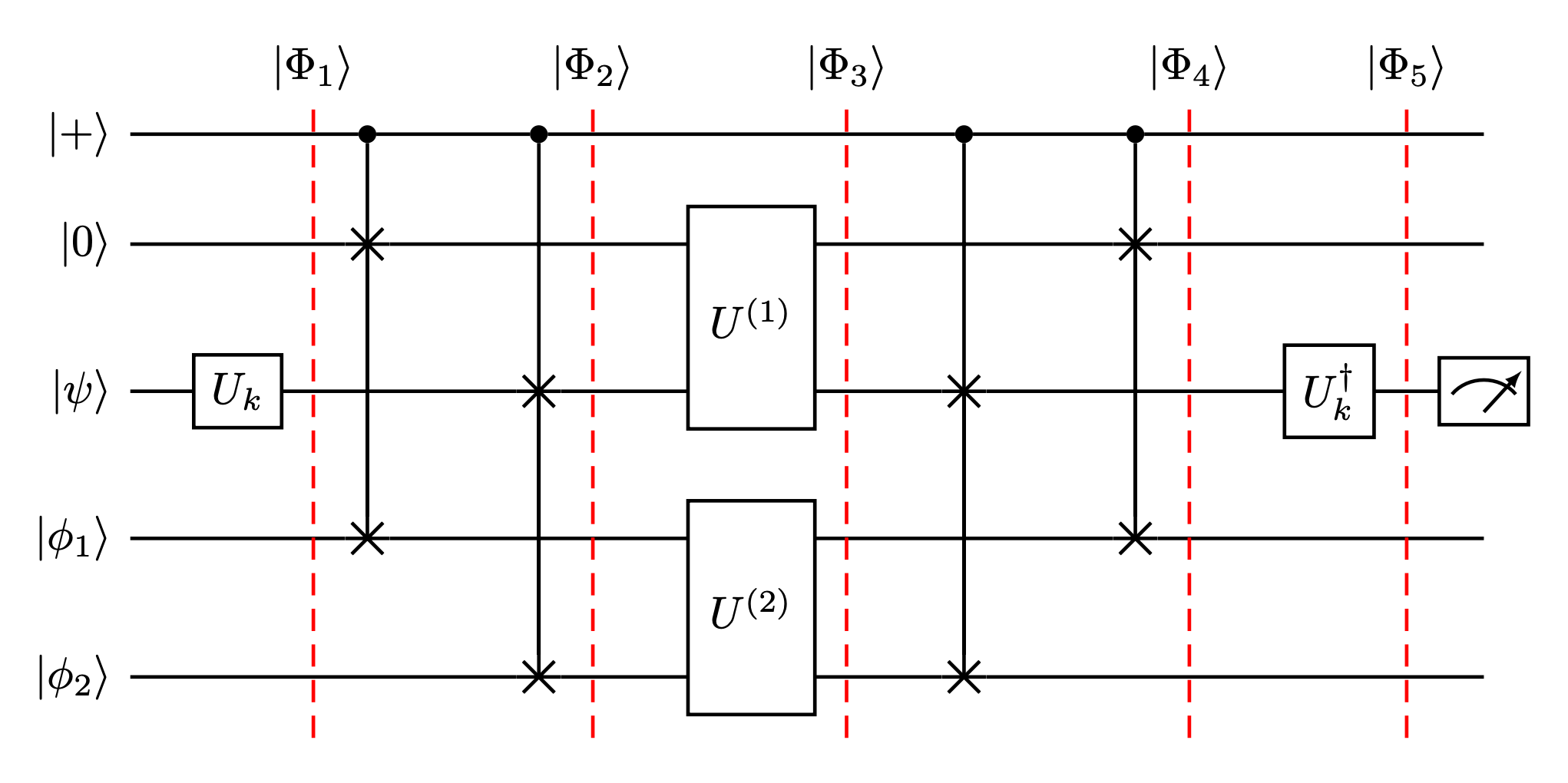}
	\caption{The quantum circuit to implement the channel $T^{(k)}_{t}$. Where $\ket{\psi}$ is the system qubit, the qubit in state $\ket{+}$ is called the ancilla lablled $a$, the qubit in $\ket{0}$ is the environment $E$ and the states $\ket{\phi_{1}}$ and $\ket{\phi_{2}}$ are arbitrary states needed for the quantum forking algorithm that implements the convex mixture of the extreme channels.}
	\label{FigConvexForking}
\end{figure*}
From the Stinespring representation and equations (\ref{eq103}) and (\ref{eq105}) we have,
\begin{equation}
	\label{eq107}
	T_{(t,1)}^{(\theta_{k})}(\rho)  =\mathrm{tr}_{E}[U^{(1)}(\ketbra{0} \otimes \rho )U^{(1)\dagger}],
	\end{equation}
	\begin{equation}
	\label{eq108}
		T_{(t,2)}^{(\theta_{k})}(\rho) =\mathrm{tr}_{E}[U^{(2)}(\ketbra{0} \otimes \rho )U^{(2)\dagger}],
	\end{equation}
where,
\begin{align}
	\label{eq109}
	U^{(1)}=\frac{1}{\sqrt{2}}\begin{pmatrix}
		ae^{-i\phi_{1}} & 0 & 0 & -c\\
		0 & d & -be^{-i\phi_{2}} & 0\\
		0 & be^{i\phi_{2}} & d & 0 \\
		c & 0 & 0 & ae^{i\phi_{1}}
	\end{pmatrix},
\end{align}
and 
\begin{align}
\label{eq110}
	U^{(2)}=\frac{1}{\sqrt{2}}\begin{pmatrix}
		ae^{i\phi_{1}} & 0 & 0 & -c\\
		0 & d & -be^{i\phi_{2}} & 0 \\
		0 & be^{-i\phi_{2}} & d & 0\\
		c & 0 & 0 & ae^{-i\phi_{1}} 
	\end{pmatrix}.
\end{align}
Now from equation (\ref{eq101}), we implement a convex sum of the channels in equations (\ref{eq107}) and (\ref{eq108}), which leads to,
\begin{align}
	\label{eq111}
	T_{t}^{(\theta_{k})}(\rho)&=\frac{1}{2}T_{(t,1)}^{(\theta_{k})}(\rho)+\frac{1}{2}T_{(t,2)}^{(\theta_{k})}(\rho),\nonumber\\
	&=\frac{1}{2}\mathrm{tr}_{E}[U^{(1)}(\ketbra{0} \otimes \rho  )U^{(1)\dagger}]+\frac{1}{2}\mathrm{tr}_{E}[U^{(2)}(\ketbra{0} \otimes \rho )U^{(2)\dagger}].\end{align}
Now to simulate the constituent channel $T_{t}^{(k)}$ we need to unitary conjugate the channel $T_{t}^{(\theta_{k})}$ bu $U_{k}$, so from equation (\ref{eq111}) we have that,
\begin{align}
\label{eq112}
	T_{t}^{(k)}(\rho)&= U_{k}^{\dagger}T_{t}^{(\theta_{k})}(U_{k}\rho U_{k}^{\dagger})U_{k},\nonumber\\
	&=U_{k}^{\dagger}(\frac{1}{2}T_{(t,1)}^{(\theta_{k})}(\rho)+\frac{1}{2}T_{(t,2)}^{(\theta_{k})}(\rho))U_{k},\nonumber\\
	&=U_{k}^{\dagger}(\frac{1}{2}\mathrm{tr}_{E}[U^{(1)}(\ketbra{0} \otimes U_{k}\rho U_{k}^{\dagger} )U^{(1)\dagger}]+\frac{1}{2}\mathrm{tr}_{E}[U^{(2)}(\ketbra{0} \otimes U_{k}\rho U_{k}^{\dagger} )U^{(2)\dagger}])U_{k},\nonumber\\
	&=\frac{1}{2}U_{k}^{\dagger}\mathrm{tr}_{E}[U^{(1)}(\ketbra{0} \otimes U_{k}\rho U_{k}^{\dagger} )U^{(1)\dagger}]U_{k}+\frac{1}{2}U_{k}^{\dagger}\mathrm{tr}_{E}[U^{(2)}(\ketbra{0} \otimes U_{k}\rho U_{k}^{\dagger} )U^{(2)\dagger}])U_{k}.
\end{align}
Equation (\ref{eq112}) tells us that we need a quantum circuit that can implement a convex sum of the two unitaries $U^{(1)}$ and $U^{(2)}$ and the unitary conjugation by $U_{k}$, for this we make use of quantum forking \cite{park2019parallel}. The theorem below shall outline how quantum forking is used to implement the convex sum of quasi-extreme channels.
\begin{theorem}
	Given a convex sum of two quantum channels as in equation (\ref{eq101}) and their Stinespring representations i.e. equations (\ref{eq107}) and (\ref{eq108}). The circuit shown in Figure \ref{FigConvexForking}. implements the convex sum of quasi extreme channels as well as unitary conjugation in equation (\ref{eq112}).
\end{theorem}

\begin{proof}
	We can show this directly by calculating the output state of the circuit in Figure \ref{FigConvexForking}. and tracing out everything else except the system qubit. We note that the input states of the qubits are as follows: $\ket{+}=\frac{1}{\sqrt{2}}(\ket{0}+\ket{1})$ is called the ancilla qubit, $\ket{\psi}$ is the state of the system, $\ket{0}$ is the state of the environment and $\ket{\phi_{1}}$, $\ket{\phi_{2}}$ are arbitrary states that are needed to perform the quantum forking algorithm. The initial state of all the qubits is:
	\begin{align}
		\label{eq113}
		\ket{\Phi_{0}}= \ket{+}\ket{0}\ket{\psi}\ket{\phi_{1}}\ket{\phi_{2}}=\frac{1}{\sqrt{2}}(\ket{0}\ket{0}\ket{\psi}\ket{\phi_{1}}\ket{\phi_{2}}+\ket{1}\ket{0}\ket{\psi}\ket{\phi_{1}}\ket{\phi_{2}}).
	\end{align}
	We now calculate $\ket{\Phi_{1}}$, we do this by applying the operator $U_{k}$ to the the system qubit $\ket{\psi}$. Now by denoting $\ket{\psi'}=U_{k}\ket{\psi}$ we have,
\begin{align}
\label{eq114}
	\ket{\Phi_{1}}=\frac{1}{\sqrt{2}}(\ket{0}\ket{0}\ket{\psi'}\ket{\phi_{1}}\ket{\phi_{2}}+\ket{1}\ket{0}\ket{\psi'}\ket{\phi_{1}}\ket{\phi_{2}}).
\end{align}
Now a control swap gates swap the states $\ket{\psi'}$ and $\ket{\phi_{1}}$ when the ancilla is in the state $\ket{1}$ and similarly the states $\ket{0}$ and $\ket{\phi_{2}}$ are swapped when the ancilla is in the state $\ket{1}$, this leads to the state,
\begin{align}
	\label{eq115}
	\ket{\Phi_{2}}=\frac{1}{\sqrt{2}}(\ket{0}\ket{0}\ket{\psi'}\ket{\phi_{1}}\ket{\phi_{2}}+\ket{1}\ket{\phi_{1}}\ket{\phi_{2}}\ket{0}\ket{\psi'}).
\end{align}
We now apply the unitary operations $U^{(1)}$ and $U^{(2)}$ on the respective qubits, this yields,
\begin{equation}
	\label{eq116}
	\ket{\Phi_{3}}=\frac{1}{\sqrt{2}}(\ket{0}U^{(1)}(\ket{0}\ket{\psi'})U^{(2)}(\ket{\phi_{1}}\ket{\phi_{2}})+\ket{1}U^{(1)}(\ket{\phi_{1}}\ket{\phi_{2}})U^{(2)}(\ket{0}\ket{\psi'})).
\end{equation}
Now we apply the controlled swaps to swap the second and third qubits with the fourth and fifth qubits when the ancilla qubit is in the state $\ket{1}$, so we have:
\begin{equation}
	\label{eq117}
	\ket{\Phi_{4}}=\frac{1}{\sqrt{2}}(\ket{0}U^{(1)}(\ket{0}\ket{\psi'})U^{(2)}(\ket{\phi_{1}}\ket{\phi_{2}})+\ket{1}U^{(2)}(\ket{0}\ket{\psi'})U^{(1)}(\ket{\phi_{1}}\ket{\phi_{2}})).
\end{equation}
Lastly we apply the operation $U_{k}^{\dagger}$ to the system qubit which yields,
\begin{equation}
	\label{eq118}
	\ket{\Phi_{5}}=\frac{1}{\sqrt{2}}(\ket{0}(\mathbb{1} \otimes U_{k}^{\dagger})U^{(1)}(\ket{0}\ket{\psi'})U^{(2)}(\ket{\phi_{1}}\ket{\phi_{2}})+\ket{1}(\mathbb{1} \otimes U_{k}^{\dagger} )U^{(2)}(\ket{0}\ket{\psi'})U^{(1)}(\ket{\phi_{1}}\ket{\phi_{2}})).
\end{equation}
We now construct the density matrix that corresponds to the state of all the qubits after performing the circuit, denoted by $\rho_{tot}=\ketbra{\Phi_{5}}$, so we have,
\begin{align}
\label{eq119}
	\rho_{tot}&=\frac{1}{2} \big( \ketbra{0}{0}\otimes (\mathbb{1} \otimes U_{k}^{\dagger} )U^{(1)}(\ket{0}\ket{\psi'})(\bra{0}\bra{\psi'})U^{(1)\dagger}(\mathbb{1} \otimes U_{k})\otimes U^{(2)}(\ket{\phi_{1}}\ket{\phi_{2}})(\bra{\phi_{1}}\bra{\phi_{2}})U^{(2)\dagger}\nonumber\\
	\nonumber\\
	&+\ketbra{0}{1}\otimes (\mathbb{1} \otimes U_{k}^{\dagger} )U^{(1)}(\ket{0}\ket{\psi'})(\bra{0}\bra{\psi'})U^{(2)\dagger}(\mathbb{1} \otimes U_{k})\otimes U^{(2)}(\ket{\phi_{1}}\ket{\phi_{2}})(\bra{\phi_{1}}\bra{\phi_{2}})U^{(1)\dagger}\nonumber\\
	\nonumber\\
	&+\ketbra{1}{0}\otimes (\mathbb{1} \otimes U_{k}^{\dagger} )U^{(2)}(\ket{0}\ket{\psi'})(\bra{0}\bra{\psi'})U^{(1)\dagger}(\mathbb{1} \otimes U_{k})\otimes U^{(1)}(\ket{\phi_{1}}\ket{\phi_{2}})(\bra{\phi_{1}}\bra{\phi_{2}})U^{(2)\dagger}\nonumber\\
	\nonumber\\
	&+ \ketbra{1}{1}\otimes (\mathbb{1} \otimes U_{k}^{\dagger} )U^{(2)}(\ket{0}\ket{\psi'})(\bra{0}\bra{\psi'})U^{(2)\dagger}(\mathbb{1} \otimes U_{k})\otimes U^{(1)}(\ket{\phi_{1}}\ket{\phi_{2}})(\bra{\phi_{1}}\bra{\phi_{2}})U^{(1)\dagger} \big).\nonumber\\
	\nonumber\\
\end{align}
From the density matrix of the total system of qubits we can now trace out the environment $E$, the ancilla $a$ and the qubits in an arbitrary state, used for the forking, $\phi_{1}$ and $\phi_{2}$. Doing this partial trace leaves us with the state of the system qubit after performing the circuit. We can do this in a few steps: We can start by tracing out the ancilla, we do this by taking the trace over the first tensor factor from the left in $\rho_{tot}$ while using the fact that $\mathrm{tr}[\ketbra{a}{b}]=\braket{a}{b}$ this leads to,
\begin{align}
	\label{eq120}
	\mathrm{tr}_{a}[\rho_{tot}]&=\frac{1}{2} \big( \braket{0}{0}\otimes (\mathbb{1} \otimes U_{k}^{\dagger})U^{(1)}(\ket{0}\ket{\psi'})(\bra{0}\bra{\psi'})U^{(1)\dagger}(\mathbb{1} \otimes U_{k})\otimes U^{(2)}(\ket{\phi_{1}}\ket{\phi_{2}})(\bra{\phi_{1}}\bra{\phi_{2}})U^{(2)\dagger}\nonumber\\
	\nonumber\\
&+\braket{0}{1}\otimes (\mathbb{1} \otimes U_{k}^{\dagger})U^{(1)}(\ket{0}\ket{\psi'})(\bra{0}\bra{\psi'})U^{(2)\dagger}(\mathbb{1} \otimes U_{k})\otimes U^{(2)}(\ket{\phi_{1}}\ket{\phi_{2}})(\bra{\phi_{1}}\bra{\phi_{2}})U^{(1)\dagger}\nonumber\\
	\nonumber\\
	&+\braket{1}{0}\otimes (\mathbb{1} \otimes U_{k}^{\dagger})U^{(2)}(\ket{0}\ket{\psi'})(\bra{0}\bra{\psi'})U^{(1)\dagger}(\mathbb{1} \otimes U_{k})\otimes U^{(1)}(\ket{\phi_{1}}\ket{\phi_{2}})(\bra{\phi_{1}}\bra{\phi_{2}})U^{(2)\dagger}\nonumber\\
	\nonumber\\
	&+\braket{1}{1}\otimes (\mathbb{1} \otimes U_{k}^{\dagger})U^{(2)}(\ket{0}\ket{\psi'})(\bra{0}\bra{\psi'})U^{(2)\dagger}(\mathbb{1} \otimes U_{k})\otimes U^{(1)}(\ket{\phi_{1}}\ket{\phi_{2}})(\bra{\phi_{1}}\bra{\phi_{2}})U^{(1)\dagger} \big).\nonumber\\
	\nonumber\\
	\end{align}

By using the orthonormality of the states $\{\ket{0},\ket{1}\}$, we can simplify equation (\ref{eq120}),
\begin{align}
	\label{eq121}
	\mathrm{tr}_{a}[\rho_{tot}]&=\frac{1}{2} \big( (\mathbb{1} \otimes U_{k}^{\dagger})U^{(1)}(\ket{0}\ket{\psi'})(\bra{0}\bra{\psi'})U^{(1)\dagger}(\mathbb{1} \otimes U_{k})\otimes U^{(2)}(\ket{\phi_{1}}\ket{\phi_{2}})(\bra{\phi_{1}}\bra{\phi_{2}})U^{(2)\dagger}\nonumber\\
	&+ (\mathbb{1} \otimes U_{k}^{\dagger})U^{(2)}(\ket{0}\ket{\psi'})(\bra{0}\bra{\psi'})U^{(2)\dagger}(\mathbb{1} \otimes U_{k})\otimes U^{(1)}(\ket{\phi_{1}}\ket{\phi_{2}})(\bra{\phi_{1}}\bra{\phi_{2}})U^{(1)\dagger} \big).\nonumber\\
\end{align}
Now let us trace out the arbitrary qubits $\phi_{1}$ and $\phi_{2}$,
\begin{align}
	\label{eq122}
	\mathrm{tr}&_{\phi_{1}+\phi_{2}+a}[\rho_{tot}]=\mathrm{tr}_{\phi_{1}+\phi_{2}}[\mathrm{tr}_{a}[\rho_{tot}]],\nonumber\\
	&=\frac{1}{2} \big( (\mathbb{1} \otimes U_{k}^{\dagger})U^{(1)}(\ket{0}\ket{\psi'})(\bra{0}\bra{\psi'})U^{(1)\dagger}(\mathbb{1} \otimes U_{k})\otimes \mathrm{tr}[U^{(2)}(\ket{\phi_{1}}\ket{\phi_{2}})(\bra{\phi_{1}}\bra{\phi_{2}})U^{(2)\dagger}]\nonumber\\
	&+(\mathbb{1} \otimes U_{k}^{\dagger})U^{(2)}(\ket{0}\ket{\psi'})(\bra{0}\bra{\psi'})U^{(2)\dagger}(\mathbb{1} \otimes U_{k})\otimes \mathrm{tr}[U^{(1)}(\ket{\phi_{1}}\ket{\phi_{2}})(\bra{\phi_{1}}\bra{\phi_{2}})U^{(1)\dagger}] \big).\nonumber\\
\end{align}
Using the fact that,
\begin{equation}
	\label{eq123}
	\mathrm{tr}[U^{(2)}(\ket{\phi_{1}}\ket{\phi_{2}})(\bra{\phi_{1}}\bra{\phi_{2}})U^{(2)\dagger}]=1=\mathrm{tr}[U^{(1)}(\ket{\phi_{1}}\ket{\phi_{2}})(\bra{\phi_{1}}\bra{\phi_{2}})U^{(1)\dagger}],
\end{equation}
we can simplify equation (\ref{eq122}),
\begin{align}
	\label{eq124}
	\mathrm{tr}_{\phi_{1}+\phi_{2}+a}[\rho_{tot}]&=\frac{1}{2} \big( (\mathbb{1} \otimes U_{k}^{\dagger})U^{(1)}(\ket{0}\ket{\psi'})(\bra{0}\bra{\psi'})U^{(1)\dagger}(\mathbb{1} \otimes U_{k})\nonumber\\
	&+(\mathbb{1} \otimes U_{k}^{\dagger})U^{(2)}(\ket{0}\ket{\psi'})(\bra{0}\bra{\psi'})U^{(2)\dagger}(\mathbb{1} \otimes U_{k})\big).
\end{align}
We can now define the density matrix $\rho_{SE}=\mathrm{tr}_{\phi_{1}+\phi_{2}+a}[\rho_{tot}]$, which denotes the state of the system qubit and environment. We can now obtain the state of the system $\rho_{S}$, after performing the circuit, by tracing out the environment. From equation (\ref{eq124}) we have,
\begin{align}
\label{eq125}
	\rho_{S}&=\mathrm{tr}_{E}[\rho_{SE}],\nonumber\\
	&=\frac{1}{2}\big(\mathrm{tr}_{E}[(\mathbb{1} \otimes U_{k}^{\dagger})U^{(1)}(\ket{0}\ket{\psi'})(\bra{0}\bra{\psi'})U^{(1)\dagger}(\mathbb{1} \otimes U_{k})]\nonumber\\
	&+\mathrm{tr}_{E}[(\mathbb{1} \otimes U_{k}^{\dagger})U^{(2)}(\ket{0}\ket{\psi'})(\bra{0}\bra{\psi'})U^{(2)\dagger}(\mathbb{1} \otimes U_{k})\big).
\end{align}
Now by using the states $\{\ket{0},\ket{1}\}$ as an orthonormal basis for the environment we can write equation (\ref{eq125}) as,
\begin{align}
	\label{eq126}
	\rho_{S}&=\frac{1}{2}\bigg( \sum_{j=0}^{1}( \bra{j} \otimes \mathbb{1})(\mathbb{1} \otimes U_{k}^{\dagger})U^{(1)}(\ket{0}\ket{\psi'})(\bra{0}\bra{\psi'})U^{(1)\dagger}(\mathbb{1} \otimes U_{k})( \ket{j} \otimes \mathbb{1})\nonumber\\
	&+\sum_{j'=0}^{1}(\bra{j'} \otimes \mathbb{1})(\mathbb{1} \otimes U_{k}^{\dagger})U^{(2)}(\ket{0}\ket{\psi'})(\bra{0}\bra{\psi'})U^{(2)\dagger}(\mathbb{1} \otimes U_{k})(  \ket{j'} \otimes \mathbb{1}) \bigg).
\end{align}
By using the fact that $(\bra{j} \otimes \mathbb{1})(\mathbb{1} \otimes U_{k}^{\dagger})=(\mathbb{1} \otimes U_{k}^{\dagger})(\bra{j} \otimes \mathbb{1})$ and $(\mathbb{1} \otimes U_{k})( \ket{j} \otimes \mathbb{1})=(\ket{j} \otimes \mathbb{1})(\mathbb{1} \otimes U_{k})$, we can write $\rho_{S}$ as,
\begin{align}
	\label{eq127}
	\rho_{S}&=\frac{1}{2}\bigg( \sum_{j=0}^{1}(\mathbb{1} \otimes U_{k}^{\dagger})( \bra{j} \otimes \mathbb{1})U^{(1)}(\ket{0}\ket{\psi'})(\bra{0}\bra{\psi'})U^{(1)\dagger}( \ket{j} \otimes \mathbb{1})(U_{k}\otimes \mathbb{1})\nonumber\\
	&+\sum_{j'=0}^{1}(\mathbb{1} \otimes U_{k}^{\dagger})(\bra{j'} \otimes \mathbb{1})U^{(2)}(\ket{0}\ket{\psi'})(\bra{0}\bra{\psi'})U^{(2)\dagger}(  \ket{j'} \otimes \mathbb{1})(U_{k}\otimes \mathbb{1}) \bigg),\nonumber\\
	&=\frac{1}{2}\big( U_{k}^{\dagger}\mathrm{tr}_{E}[U^{(1)}(\ket{0}\ket{\psi'})(\bra{0}\bra{\psi'})U^{(1)\dagger}]U_{k}+U_{k}^{\dagger}\mathrm{tr}_{E}[U^{(2)}(\ket{0}\ket{\psi'})(\bra{0}\bra{\psi'})U^{(2)\dagger}]U_{k}\big).\end{align}
By observing that $(\ket{0}\ket{\psi'})(\bra{0}\bra{\psi'})=\ketbra{0}{0} \otimes \ketbra{\psi'}{\psi'} $, and from the definition of $\ket{\psi'}$ we have that $\ketbra{0}{0} \otimes \ketbra{\psi'}{\psi'}=\ketbra{0} \otimes U_{k}\ketbra{\psi}U_{k}^{\dagger}$. If we define $\rho=\ketbra{\psi}$, we can write $(\ket{0}\ket{\psi'})(\bra{0}\bra{\psi'})=\ketbra{0} \otimes U_{k} \rho U_{k}^{\dagger}$, this allows us to write equation (\ref{eq127}) as,
\begin{align}
\label{eq128}
	\rho_{S}&=\frac{1}{2}\big( U_{k}^{\dagger}\mathrm{tr}_{E}[U^{(1)}(\ketbra{0} \otimes U_{k}\rho U_{k}^{\dagger} )U^{(1)\dagger}]U_{k}+U_{k}^{\dagger}\mathrm{tr}_{E}[U^{(2)}(\ketbra{0} \otimes U_{k}\rho U_{k}^{\dagger} )U^{(2)\dagger}]U_{k}\big),\nonumber\\
	&=\frac{1}{2}\big(U_{k}^{\dagger}T_{(t,1)}^{(\theta_{k})}(U_{k}\rho U_{k}^{\dagger})U_{k} +U_{k}^{\dagger}T_{(t,2)}^{(\theta_{k})}(U_{k}\rho U_{k}^{\dagger})U_{k}\big),\nonumber\\
	&=U_{k}^{\dagger}\big( \frac{1}{2}T_{(t,1)}^{(\theta_{k})}(U_{k}\rho U_{k}^{\dagger})+\frac{1}{2}T_{(t,2)}^{(\theta_{k})}(U_{k}\rho U_{k}^{\dagger})\big)U_{k},\nonumber\\
	&=U_{k}^{\dagger}T_{t}^{(\theta_{k})}(U_{k}\rho U_{k}^{\dagger})U_{k},\nonumber\\
	&=T_{t}^{(k)}(\rho).
\end{align}
We get the second equality in equation (\ref{eq128}) from the Stinespring representation of channels $T_{(t,1)}^{(\theta_{k})}$ and $T_{(t,2)}^{(\theta_{k})}$ and the last equality follows from equation (\ref{eq112}). Hence we have shown that the circuit in Figure \ref{FigConvexForking}. implements the convex sum of quasi extreme channels and unitary conjugation in equation (\ref{eq112}), which completes the proof.
\end{proof}
Now that we have a circuit that can implement the constituent channel $T_{t}^{(k)}$ i.e. Figure \ref{FigConvexForking}. We need to decompose the unitary operations $U^{(1)}$ and $U^{(2)}$ into single qubit gates and controlled not (CNOT) gates. Before we can do this we observe that, $\frac{1}{2}(a^{2}+c^{2})=1$ and $\frac{1}{2}(b^{2}+d^{2})=1$ so we can define,
\begin{align}
\label{eq129}
	\cos(\alpha)&=\frac{1}{\sqrt{2}}a, && \sin(\alpha)=\frac{1}{\sqrt{2}}c,\nonumber\\
	\cos(\beta)&=\frac{1}{\sqrt{2}}b, && \sin(\beta)=\frac{1}{\sqrt{2}}d.
\end{align}
This allows us to rewrite the unitary operations $U^{(1)}$ and $U^{(2)}$ as,
\begin{equation}
	\label{eq130}
	U^{(1)}=\begin{pmatrix}
		e^{-i\phi_{1}}\cos(\alpha) & 0 & 0 & -\sin(\alpha)\\
		0 & \sin(\beta) & -e^{-i\phi_{2}}\cos(\beta) & 0 \\
		0 & e^{i\phi_{2}}\cos(\beta) & \sin(\beta) & 0\\
		\sin(\alpha) & 0 & 0 & e^{i\phi_{1}}\cos(\alpha)
	\end{pmatrix},
\end{equation}
and
\begin{equation}
	\label{eq131}
	U^{(2)}=\begin{pmatrix}
		e^{i\phi_{1}}\cos(\alpha) & 0 & 0 & -\sin(\alpha)\\
		0 & \sin(\beta) & -e^{i\phi_{2}}\cos(\beta) & 0 \\
		0 & e^{-i\phi_{2}}\cos(\beta) & \sin(\beta) & 0 \\
		\sin(\alpha) & 0 & 0 & e^{-i\phi_{1}}\cos(\alpha)
	\end{pmatrix}.
\end{equation}
Now let us start by decomposing the unitary matrix $U^{(1)}$, we observe that,
\begin{align}
	\label{eq132}
	U^{(1)}=U_{A}^{(1)}U_{B}^{(1)},
\end{align}
where
\begin{align}
	\label{eq133}
	U_{A}^{(1)}=\begin{pmatrix}
		e^{-i\phi_{1}}\cos(\alpha) & 0 & 0 & -\sin(\alpha)\\
		0 & 1 & 0 & 0 \\
		0 & 0 & 1 & 0\\
		\sin(\alpha) & 0 & 0 & e^{i\phi_{1}}\cos(\alpha)
	\end{pmatrix},
\end{align}
and 
\begin{align}
	\label{eq134}
	U_{B}^{(1)}=\begin{pmatrix}
		1 & 0 & 0 &0\\
		0 & \sin(\beta) & -e^{-i\phi_{2}}\cos(\beta) & 0 \\
		0 & e^{i\phi_{2}}\cos(\beta) & \sin(\beta) & 0\\
		0 & 0 & 0 & 1
	\end{pmatrix}.
\end{align}
Now using the methods outlined in Appendix B we can find quantum circuits for the operations $U_{A}^{(1)}$ and $U_{B}^{(2)}$. First we define the following $2 \times 2$ matrices, which will assist us in the decomposition,
\begin{align}
	\label{eq135}
	\Tilde{U}_{A}^{(1)}=\begin{pmatrix}
		e^{-i\phi_{1}}\cos(\alpha) & -\sin(\alpha)\\
		\sin(\alpha) & e^{i\phi_{1}}\cos(\alpha)
	\end{pmatrix}, && \Tilde{U}_{B}^{(1)}=\begin{pmatrix}
		\sin(\beta) & -e^{-i\phi_{2}}\cos(\beta)\\
		e^{i\phi_{2}}\cos(\beta) & \sin(\beta)
	\end{pmatrix}.
\end{align}
Now for the unitary matrix $U_{A}^{(1)}$ we can write,
\begin{align}
\label{eq136}
	U_{A}^{(1)}&=\begin{pmatrix}
		0 & 1 & 0 & 0\\
		1 & 0 & 0 & 0\\
		0 & 0 & 1 & 0\\
		0 & 0 & 0 & 1
	\end{pmatrix}\begin{pmatrix}
		1 & 0 & 0 & 0 \\
		0 & e^{-i\phi_{1}}\cos(\alpha) & 0 & -\sin(\alpha)\\
		0 & 0 & 1 & 0\\
		0 & \sin(\alpha) & 0 & e^{i\phi_{1}}\cos(\alpha)
	\end{pmatrix}\begin{pmatrix}
		0 & 1 & 0 & 0\\
		1 & 0 & 0 & 0\\
		0 & 0 & 1 & 0\\
		0 & 0 & 0 & 1
	\end{pmatrix}\nonumber\\
	\nonumber\\
	&=(\ketbra{0}\otimes X +\ketbra{1}\otimes \mathbb{1})(\mathbb{1}\otimes \ketbra{0}+\Tilde{U}_{A}^{(1)}\otimes \ketbra{1})(\ketbra{0}\otimes X +\ketbra{1}\otimes \mathbb{1}).
\end{align}
From equation (\ref{eq136}) we can find the circuit that corresponds to $U_{A}^{(1)}$, and it is shown in Figure \ref{FigcircUAB} (a). Now using the same method above we can find a circuit that implements the operation $U_{B}^{(1)}$,
\begin{align}
	\label{eq137}
	U_{B}^{(1)}&=\begin{pmatrix}
		0 & 1 & 0 & 0\\
		1 & 0 & 0 & 0\\
		0 & 0 & 1 & 0\\
		0 & 0 & 0 & 1
	\end{pmatrix}\begin{pmatrix}
	   \sin(\beta) & 0 & -e^{-i\phi_{2}}\cos(\beta) & 0 \\
		0 & 1 & 0 &0\\
	    e^{i\phi_{2}}\cos(\beta) & 0 & \sin(\beta) & 0\\
		0 & 0 & 0 & 1
	\end{pmatrix}\begin{pmatrix}
		0 & 1 & 0 & 0\\
		1 & 0 & 0 & 0\\
		0 & 0 & 1 & 0\\
		0 & 0 & 0 & 1
	\end{pmatrix},\nonumber\\
	\nonumber\\
	&=(\ketbra{0}\otimes X +\ketbra{1}\otimes \mathbb{1})(\Tilde{U}_{B}^{(1)}\otimes \ketbra{0}+\mathbb{1}\otimes \ketbra{1})(\ketbra{0}\otimes X +\ketbra{1}\otimes \mathbb{1}).
\end{align}
The circuit that implements $U_{B}^{(1)}$ can be seen in Figure \ref{FigcircUAB} (b). Now we can put these circuits together to get the circuit that implements the unitary operation $U^{(1)}$, this can be seen in Figure \ref{FigcircU1}. Since $(\ketbra{0}\otimes X +\ketbra{1} \otimes \mathbb{1})^{2}=\mathbb{1}$, we drop the two open CNOT's in between the circuits for $U_{B}^{(1)}$ and $U_{A}^{(1)}$, and get the simplified circuit for $U^{(1)}$ in Figure \ref{FigcircU1}.
\begin{figure*}
	\centering
	\includegraphics[scale=0.35]{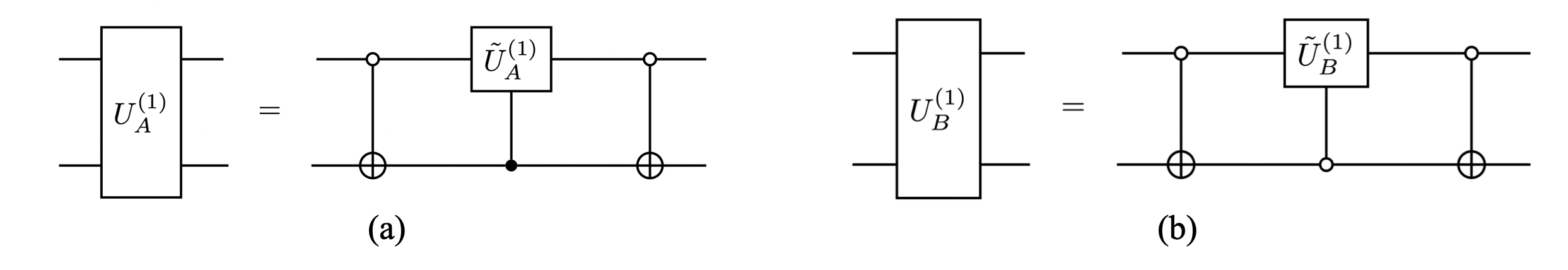}
	\caption{Circuit that implements the operators $U_{A}^{(1)}$ and $U_{A}^{(2)}$.}
	\label{FigcircUAB}
\end{figure*}

\begin{figure*}
	\centering
	\includegraphics[scale=0.35]{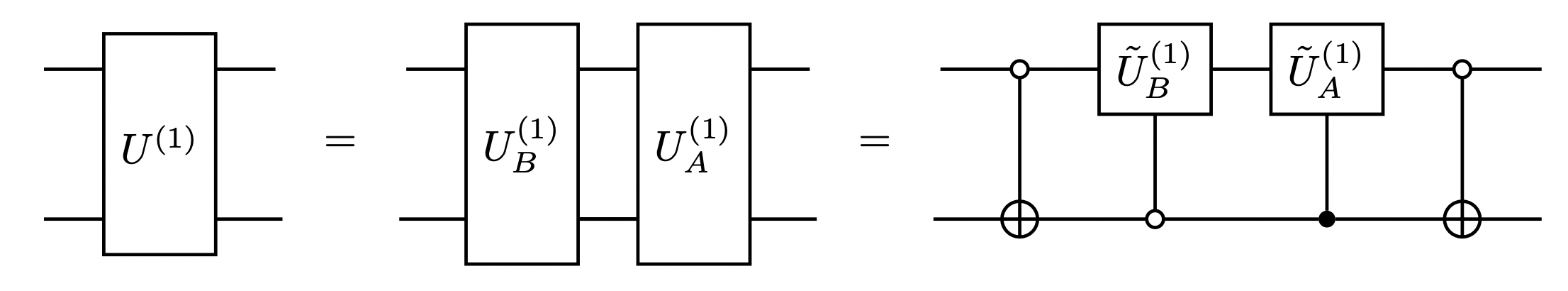}
	\caption{Circuit that implements the opertation $U^{(1)}$.}
	\label{FigcircU1}
\end{figure*}
We now want to decompose the controlled $\Tilde{U}_{A}^{(1)}$ operation into single qubit gates and CNOT's. We use the results from Appendix B to decompose this operation as follows. Since $\mathrm{det}(\Tilde{U}_{A}^{(1)})=1$, this implies that $\Tilde{U}_{A}^{(1)}\in \mathrm{SU}(2)$, them by Lemma B.2. we have,
\begin{align}
	\label{eq138}
	\Tilde{U}_{A}^{(1)}&=\begin{pmatrix}
		e^{\frac{-i\phi_{1}}{2}} & 0\\
		0 & e^{\frac{i\phi_{1}}{2}}
	\end{pmatrix}\begin{pmatrix}
		\cos(\alpha) & -\sin(\alpha)\\
		\sin(\alpha) & \cos(\alpha)
	\end{pmatrix}\begin{pmatrix}
		e^{\frac{-i\phi_{1}}{2}} & 0\\
		0 & e^{\frac{i\phi_{1}}{2}}
	\end{pmatrix},\nonumber\\
	\nonumber\\
	&=R_{z}(\phi_{1})R_{y}(2\alpha)R_{z}(\phi_{1}).
	\end{align}
Now from Lemma B.4. we can decompose the controlled $\Tilde{U}_{A}^{(1)}$ operation into single qubit and CNOT gates as seen in Figure \ref{FigcircCUA1}.\\
\begin{figure*}
	\centering
	\includegraphics[scale=0.3]{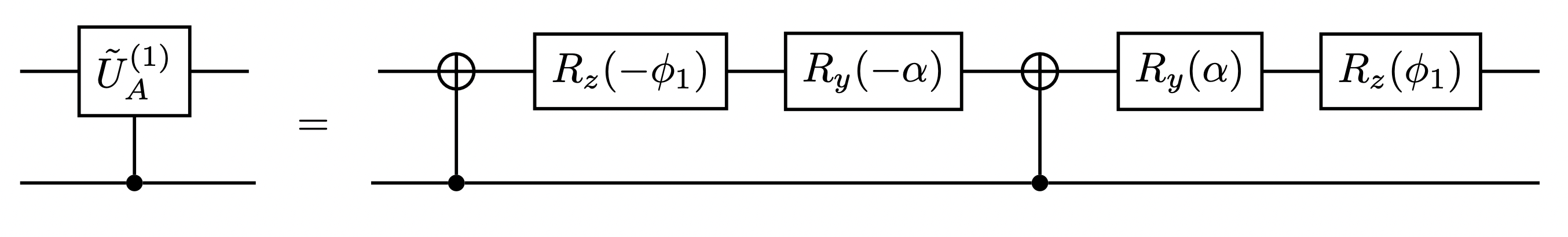}
	\caption{Circuit that implements the controlled $\Tilde{U}_{A}^{(1)}$operation.}
	\label{FigcircCUA1}
\end{figure*}
We can now decompose the open controlled $\Tilde{U}_{B}^{(1)}$ operation into single qubit gates and CNOT's. This is done in the same way as for the controlled $\Tilde{U}_{A}^{(1)}$. Since $\mathrm{det}(\Tilde{U}_{B}^{(1)})=1$, this implies that $\Tilde{U}_{B}^{(1)}\in \mathrm{SU}(2)$, them by Lemma B.2. we have,
\begin{align}
	\label{eq139}
	\Tilde{U}_{B}^{(1)}&=\begin{pmatrix}
		e^{\frac{-i\phi_{2}}{2}} & 0\\
		0 & e^{\frac{i\phi_{2}}{2}}
	\end{pmatrix}\begin{pmatrix}
		\sin(\beta) & -\cos(\beta)\\
		\cos(\beta) & \sin(\beta)
	\end{pmatrix}\begin{pmatrix}
		e^{\frac{i\phi_{2}}{2}} & 0\\
		0 & e^{\frac{-i\phi_{2}}{2}}
	\end{pmatrix},\nonumber\\
	\nonumber\\
	&=\begin{pmatrix}
		e^{\frac{-i\phi_{2}}{2}} & 0\\
		0 & e^{\frac{i\phi_{2}}{2}}
	\end{pmatrix}\begin{pmatrix}
		\cos(\frac{\pi}{2}-\beta) & -\sin(\frac{\pi}{2}-\beta)\\
		\sin(\frac{\pi}{2}-\beta) &\cos(\frac{\pi}{2}-\beta)
	\end{pmatrix}\begin{pmatrix}
		e^{\frac{i\phi_{2}}{2}} & 0\\
		0 & e^{\frac{-i\phi_{2}}{2}}
	\end{pmatrix},\nonumber\\
	\nonumber\\
	&=R_{z}(\phi_{2})R_{y}(2(\frac{\pi}{2}-\beta))R_{z}(-\phi_{2}).
\end{align}
Now from Lemma B.4. we have the circuit for the open controlled $\Tilde{U}_{B}^{(1)}$ operation in Figure \ref{FigcircCUB1}. Let us now decompose the unitary operation $U^{(2)}$. From equation (\ref{eq131}) we observe that,
\begin{figure*}
	\centering
	\includegraphics[scale=0.3]{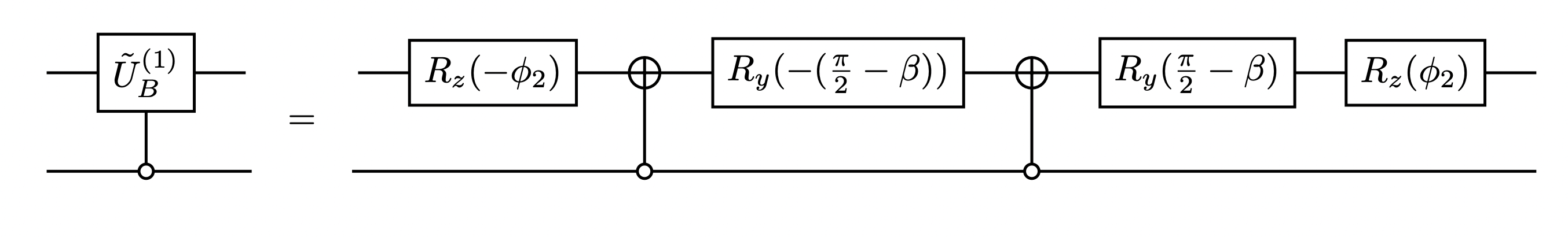}
	\caption{Circuit that implements the open controlled $\Tilde{U}_{B}^{(1)}$operation.}
	\label{FigcircCUB1}
\end{figure*}
\begin{align}
	\label{eq140}
	U^{(2)}=U_{A}^{(2)}U_{B}^{(2)},
\end{align}
where
\begin{align}
	\label{eq141}
	U_{A}^{(2)}=\begin{pmatrix}
		e^{i\phi_{1}}\cos(\alpha) & 0 & 0 & -\sin(\alpha)\\
		0 & 1 & 0 & 0 \\
		0 & 0 & 1 & 0 \\
		\sin(\alpha) & 0 & 0 & e^{-i\phi_{1}}\cos(\alpha)
	\end{pmatrix},
\end{align}
\begin{align}
	\label{eq142}
	U_{B}^{(2)}=\begin{pmatrix}
		1 & 0 & 0 & 0\\
		0 & \sin(\beta) & -e^{i\phi_{2}}\cos(\beta) & 0 \\
		0 & e^{-i\phi_{2}}\cos(\beta) & \sin(\beta) & 0 \\
		0 & 0 & 0 & 1
	\end{pmatrix}.
\end{align}
We can now define the following matrices,
\begin{align}
	\label{eq143}
	\Tilde{U}_{A}^{(2)}=\begin{pmatrix}
		e^{i\phi_{1}}\cos(\alpha) & -\sin(\alpha)\\
		\sin(\alpha) & e^{-i\phi_{1}}\cos(\alpha)
	\end{pmatrix}, && \Tilde{U}_{B}^{(2)}=\begin{pmatrix}
		\sin(\beta) & -e^{i\phi_{2}}\cos(\beta)\\
		e^{-i\phi_{2}}\cos(\beta) & \sin(\beta)
	\end{pmatrix}.
\end{align}
From the matrices in equation (\ref{eq143}), we can now decompose $U_{A}^{(2)}$ as follows,
\begin{align}
	\label{eq144}
	U_{A}^{(2)}=(\ketbra{0}\otimes X +\ketbra{1}\otimes \mathbb{1})(\mathbb{1} \otimes \ketbra{0} + \Tilde{U}_{A}^{(2)} \otimes \ketbra{1})(\ketbra{0}\otimes X +\ketbra{1}\otimes \mathbb{1}).
\end{align}
Figure \ref{FigcircUAB2}. (a) shows the circuit that implements $U_{A}^{(2)}$. Now let us decompose the unitary operation $U_{B^{(2)}}$,
\begin{figure*}
	\centering
	\includegraphics[scale=0.35]{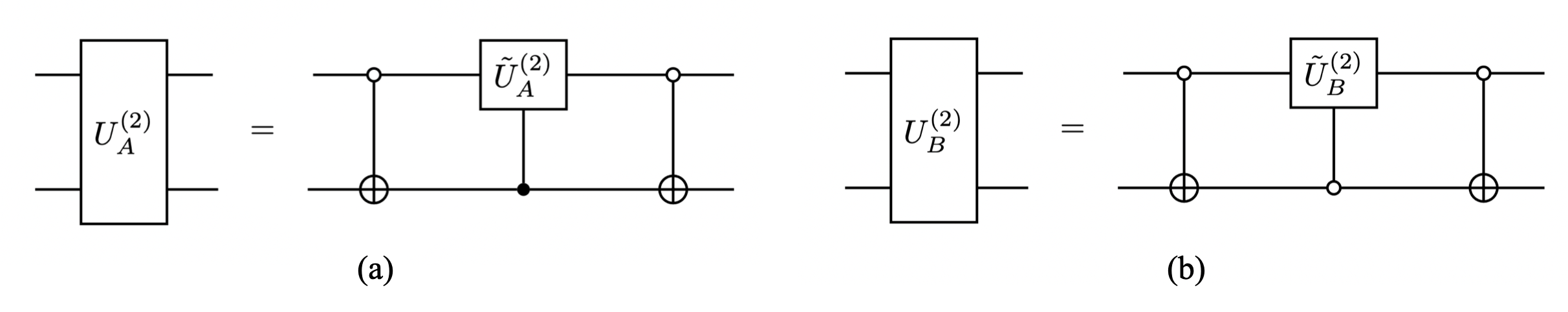}
	\caption{Circuit that implements the $U_{A}^{(2)}$ and $U_{B}^{(2)}$ unitary operations.}
	\label{FigcircUAB2}
\end{figure*}
\begin{align}
	\label{eq145}
	U_{B}^{(2)}=(\ketbra{0}\otimes X +\ketbra{1}\otimes \mathbb{1})(\mathbb{1} \otimes \ketbra{1}+\Tilde{U}_{B}^{(2)}\otimes \ketbra{0})(\ketbra{0}\otimes X +\ketbra{1}\otimes \mathbb{1}).
\end{align}
\begin{figure*}
	\centering
	\includegraphics[scale=0.34]{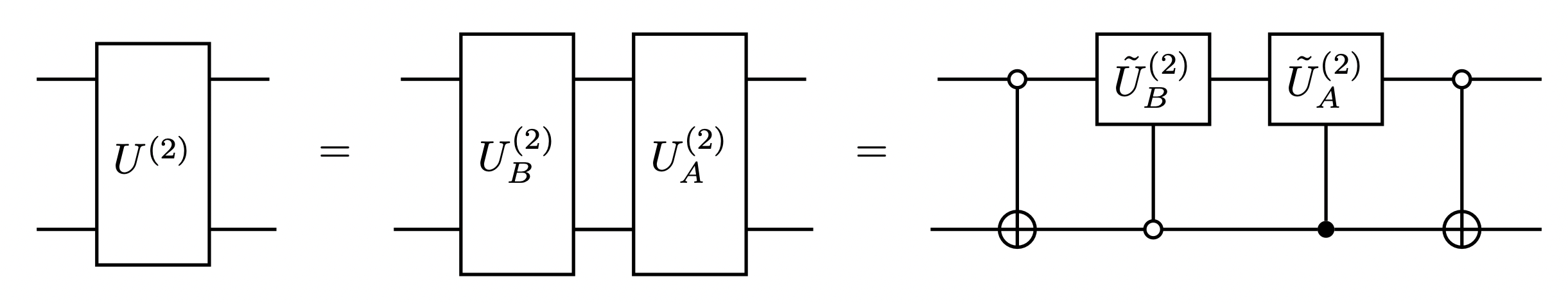}
	\caption{Circuit that implements $U^{(2)}$.}
	\label{FigcircU2}
\end{figure*}
Figure \ref{FigcircUAB2}. (b) shows the circuit that implements $U_{B}^{(2)}$. We can put this together and make the same simplification to the circuit as we did for the unitary operation $U^{(1)}$, where the two open controlled CNOT's are replaced by the identity operation. Now Figure \ref{FigcircU2}. shows the circuit that implements the unitary operation $U^{(2)}$.\\
We can now decompose the controlled $\Tilde{U}_{A}^{(2)}$ operation and the open controlled $\Tilde{U}_{B}^{(2)}$  operation using the methods from Appendix B in a similar way to the decomposition used above for $U^{(1)}$. We obtain the circuit for the controlled $\Tilde{U}_{A}^{(2)}$ operation, which is shown in Figure \ref{FigcircCUA2} and the circuit that implements the open control $\Tilde{U}_{B}^{(2)}$ operation is shown in Figure \ref{FigcircCUB2}.\\\
\begin{figure*}
	\centering
	\includegraphics[scale=0.34]{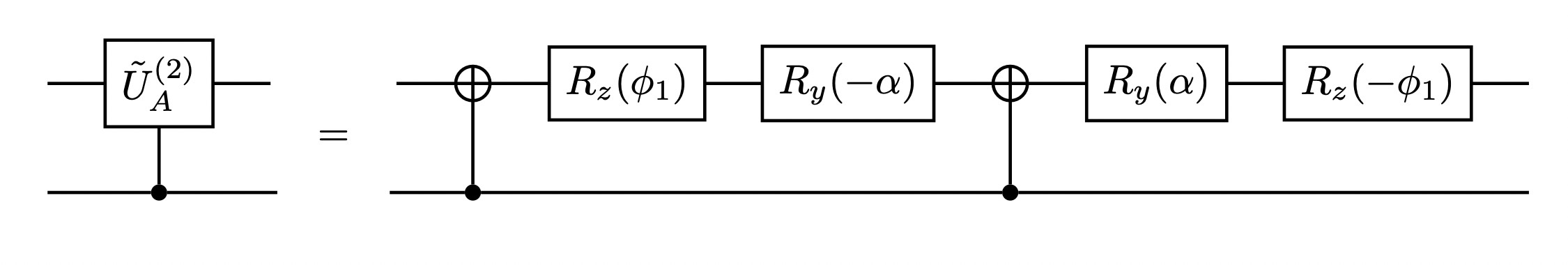}
	\caption{Circuit that implements controled $\Tilde{U}_{A}^{(2)}$.}
	\label{FigcircCUA2}
\end{figure*}
\begin{figure*}[h]
	\centering
	\includegraphics[scale=0.34]{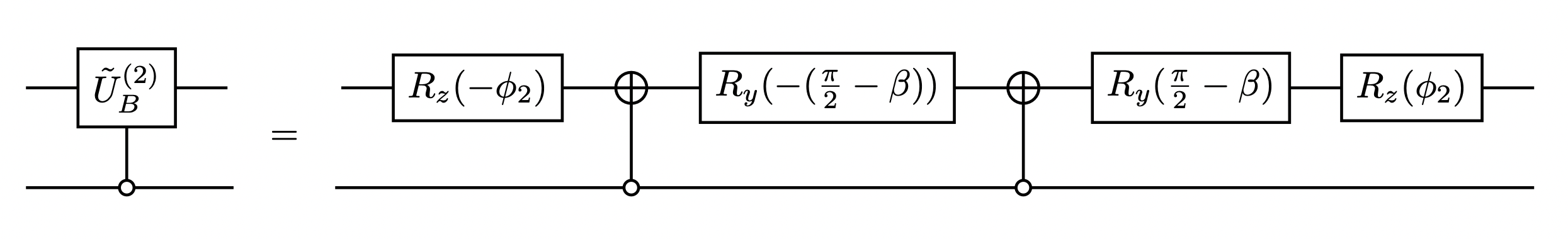}
	\caption{Circuit that implements open control $\Tilde{U}_{B}^{(2)}$.}
	\label{FigcircCUB2}
\end{figure*}
This completes the decomposition of the operations $U^{(1)}$ and $U^{(2)}$ into single qubit and CNOT gates. We can now use these circuits in the the circuit shown in Figure \ref{FigConvexForking} to simulate the constituent channel $T_{t}^{(k)}$.

\section{Summary}

We now combine the results in the previous sections and state the full algorithm for the simulation of a single qubit open quantum system, which requires $O(\frac{(\Lambda t)^{3/2}}{\epsilon^{1/2}})$ gates as a solution to the problem defined in Section 2:

(1) Given the generator $\mathcal{L}$ as defined in Section 3, the spectral decomposition of the GKS matrix $A$ yields the decomposition of $\mathcal{L}$,
\begin{align}
	\mathcal{L}=\sum_{k=0}^{3}\lambda_{k}\mathcal{L}_{k},
\end{align}
with $C_{k} \in \mathrm{SO(3)}$ and $\theta_{K}\in [-\frac{\pi}{4},\frac{\pi}{4}]$ specifying the decomposition so that,
\begin{align}
	A_{k}=C_{k}^{T}A(\theta_{k})C_{k}.
\end{align}
for all k=1,2,3. 
(2) Next we choose a precision $\epsilon \geq 0$ so that we can calculate the number of implementations $N$ of $S_{2}(\tau)$, with 
\begin{align}
	N\geq \frac{(4t\Lambda)^{3/2}}{(3\epsilon)^{1/2}}.
\end{align}
(3) Implement $S_{2}(\tau)$ $N$ times using,
\begin{align}
	 T_{t}^{(k)}(\rho)=(\mathcal{U}_{k}^{\dagger}T_{t}^{(\theta_{k})}\mathcal{U}_{k})(\rho)=U_{k}^{\dagger}[T_{t}^{(\theta_{k})}(U_{k}\rho U_{k}^{\dagger})]U_{k}.
\end{align}
where $\lambda_{k}$, $\Lambda$ and $\tau$ have been incorporated into $t'$. $U_{k}$ is obtained from $C_{k}$ as per Section 3, and $T^{(\theta_{k})}_{t'}$ is implemented via quantum forking as per Section 5. One should note that to implement products of $T^{(k)}_{t}$ we use the same circuit as in Fig. \ref{FigConvexForking} but for each implementation of $T^{(k)}_{t}$ we use a new qubit for the environment and arbitrary states $\ket{\phi_{1}}$ and $\ket{\phi_{2}}$. 

Now that we have summarised the full algorithm, one should be aware of the implication this has for simulating Markovian open quantum systems. This algorithm provides a framework for digitally simulating any single qubit Markovian open quantum system, up to a high enough precision allowing us to study the dynamics of systems for which their master equation cannot be analytically integrated. 

This tutorial set out to pedagogically introduce the tools and techniques needed to construct algorithms for the simulation of Markovian open quantum systems by looking at an algorithm for the simplest case of the simulation of a single qubit open quantum system. Of course, the reader is urged to use the ideas and techniques taught here to study the more general cases of simulating Markovian open quantum systems of arbitrary and finite dimension \cite{sweke2015universal} as well as the case of simulating non Markovian open quantum systems \cite{sweke2016digital}.

In light of the results presented in this tutorial one should also mention the possible avenues of extension. The first major issue that was recognised in developing algorithms to simulate open quantum systems and an issue that plagues this algorithm as well, is the fact that we can only approximate our channel $T_{t}$ up to second order with SLT product formulas so as not to violate complete positivity. This constraint arises from the fact that you cannot use higher order SLT product formulas for coefficients that are all positive, as proven by Suzuki in \cite{suzuki1990fractal}. Given this information one could look into methods for approximating the channel $T_{t}$ to higher orders that do not violate complete positivity. Another area of extension would be with regards to the simulation of Markovian open quantum systems with arbitrary dimension \cite{sweke2015universal}, where on could possibly tackle the open problem of convexly decomposing the universal semigroup of arbitrary dimensional channels into extreme channels \cite{ruskai2007some}.

\section*{Acknowledgements}
This work is based upon research supported by
the National Research Foundation of the Republic of
South Africa. Support from the NICIS (National Integrated Cyber Infrastructure System) e-research grant
QICSA is kindly acknowledged. We would like to thank
Ms. S. M. Pillay for her assistance in proofreading the
manuscript.

%\nocite{*}

\bibliographystyle{IEEEtran}
\bibliography{References.bib}
%\bibliographystyle{plain}
% Produces the bibliography via BibTeX.

\appendix

\section{Quantum Circuit Decomposition}
We need some useful results from \cite{barenco1995elementary}, for the decomposition of the unitary matrices corresponding to the quasi-extreme channels. We also make use of the quantum computing notation where we shall denote the Pauli matrices as:\\
\begin{align}
\label{eqB1}
	\sigma_{1}=X, && \sigma_{2}=Y, && \sigma_{3}=Z.
\end{align}
\textbf{Definition B.1.} We define the following unitary matrices,
\begin{align}
\label{eqB2}
	R_{y}(\theta)& \equiv \exp(-i\frac{\theta}{2}Y)=\cos(\frac{\theta}{2})\mathbb{1}-i\sin(\frac{\theta}{2})Y=\begin{pmatrix}
		\cos(\frac{\theta}{2}) & -\sin(\frac{\theta}{2})\\
		\\
		\sin(\frac{\theta}{2}) & \cos(\frac{\theta}{2})
	\end{pmatrix},\\
	\nonumber\\
	R_{z}(\theta)& \equiv \exp(-i\frac{\theta}{2}Z)=\cos(\frac{\theta}{2})\mathbb{1}-i\sin(\frac{\theta}{2})Z=\begin{pmatrix}
		e^{-i\frac{\theta}{2}} & 0\\
		0 & e^{i\frac{\theta}{2}}
	\end{pmatrix},\\
	\nonumber\\
	P(\delta)& \equiv \begin{pmatrix}
		e^{i\delta} & 0\\
		0 & e^{i\delta}\\
	\end{pmatrix}.
\end{align} 

We shall make use of the following lemmas to decompose any $2 \times 2$ unitary matrix.\\
\textbf{Lemma B.2.} Suppose $U$ is a $2 \times 2$ unitary matrix. Then there exists $\alpha,\beta,\delta,\theta \in \mathbb{R}$ such that,
\begin{equation}
\label{eqB3}
U=P(\delta)R_{z}(\alpha)R_{y}(\theta)R_{z}(\beta).
\end{equation}
\begin{proof}
	Since $U$ is unitary its rows and columns are orthonormal, from which it follows that there exists $\alpha,\beta,\delta,\theta \in \mathbb{R}$ such that,
	\begin{equation}
		\label{eqB4}
		U=\begin{pmatrix}
			e^{i(\delta-\frac{\alpha}{2}-\frac{\beta}{2})}\cos(\frac{\theta}{2}) & -e^{i(\delta-\frac{\alpha}{2}+\frac{\beta}{2})}\sin(\frac{\theta}{2})\\
			e^{i(\delta+\frac{\alpha}{2}-\frac{\beta}{2})}\sin(\frac{\theta}{2}) & e^{i(\delta+\frac{\alpha}{2}+\frac{\beta}{2})}\cos(\frac{\theta}{2})
		\end{pmatrix}.
	\end{equation}
	Now equation (\ref{eqB3}) follows directly from the definitions of $R_{y},R_{z},P$ and matrix multiplication.
\end{proof}
\textbf{Lemma B.3.} Suppose $U$ is a $2\times 2$ unitary matrix. Then there exists $2 \times 2$ matrices $A,B,C$ such that $ABC=\mathbb{1}$ and $U=P(\delta)AXBXC$.\\
\begin{proof}
	First we observe that, $XYX=-Y$ which implies that,\\
	\begin{align}
		\label{eqB5}
		XR_{y}(\theta)X&=X(\cos(\frac{\theta}{2})\mathbb{1}-i\sin(\frac{\theta}{2})Y)X\nonumber\\
		&=\cos(\frac{\theta}{2})\mathbb{1}-i\sin(\frac{\theta}{2})XYX\nonumber\\
		&=\cos(\frac{\theta}{2})\mathbb{1}+i\sin(\frac{\theta}{2})Y\nonumber\\
		&=R_{y}(-\theta).
	\end{align}
Now if we set $A \equiv R_{z}(\alpha)R_{y}(\frac{\theta}{2})$, $B \equiv R_{y}(-\frac{\theta}{2})R_{z}(\frac{-(\beta +\alpha )}{2})$ and $C \equiv R_{z}(\frac{(\beta -\alpha)}{2})$, we have,
\begin{align}
	\label{eqB6}
	ABC&=R_{z}(\alpha)R_{y}\left( \frac{\theta}{2}\right)R_{y}\left(-\frac{\theta}{2}\right)R_{z}\left(-\frac{\beta +\alpha }{2}\right)R_{z}\left(\frac{\beta -\alpha}{2}\right)\nonumber\\
	&=R_{z}(\alpha)R_{y}\left(\frac{\theta}{2}-\frac{\theta}{2}\right)R_{z}\left(\frac{-\beta -\alpha }{2}+\frac{\beta -\alpha}{2}\right)\nonumber\\
	&=R_{z}(\alpha)R_{y}(0)R_{z}(-\alpha)\nonumber\\
	&=R_{z}(\alpha)R_{z}(-\alpha)\nonumber\\
	&=\mathbb{1}
\end{align}
Since $X^{2}=\mathbb{1}$ and from equation (\ref{eqB5}) we see that,
\begin{align}
	\label{eqB7}
	XBX&=XR_{y}\left(-\frac{\theta}{2}\right)R_{z}\left(-\frac{(\beta +\alpha )}{2}\right)X\nonumber\\
	&=XR_{y}\left(-\frac{\theta}{2}\right)XXR_{z}\left(-\frac{(\beta +\alpha )}{2}\right)X\nonumber\\
	&=R_{y}\left(\frac{\theta}{2}\right)R_{z}\left(\frac{(\beta +\alpha )}{2}\right).
\end{align}
Thus,
\begin{align}
	\label{eqB8}
	AXBXC&=R_{z}(\alpha)R_{y}\left(\frac{\theta}{2}\right)XR_{y}\left(-\frac{\theta}{2}\right)R_{z}\left(-\frac{(\beta +\alpha )}{2}\right)XR_{z}\left(\frac{\beta -\alpha}{2}\right)\nonumber\\
	&=R_{z}(\alpha)R_{y}\left(\frac{\theta}{2}\right)R_{y}\left(\frac{\theta}{2}\right)R_{z}\left(\frac{(\beta +\alpha )}{2}\right)R_{z}\left(\frac{(\beta -\alpha)}{2}\right)\nonumber\\
	&=R_{z}(\alpha)R_{y}(\theta)R_{z}(\beta).
\end{align}
Now by Lemma B.2. we have that,
\begin{equation}
	\label{eqB9}
	U=P(\delta)AXBXC,
\end{equation}
as required.
\end{proof}
Now we require a way to decompose a controlled unitary operation. We have the following lemma which tells us how to do this.\\
\\
\textbf{Lemma B.4.} For any $2 \times 2$ unitary matrix $U$, a controlled unitary gate $c_{1}(U)$ can be simulated with the quantum circuit in Figure \ref{FigCU}. with $A,B,C \in \mathrm{SU}(2)$, if and only if $U \in \mathrm{SU}(2)$. \\
\begin{figure}
\centering
	\includegraphics[scale=0.2]{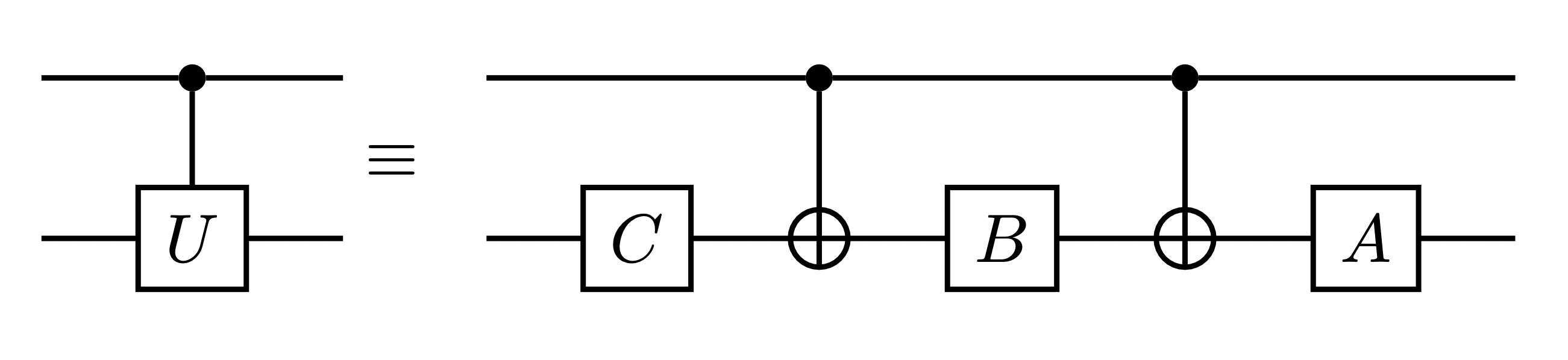}
	\caption{Quantum circuit implementing the controlled unitary gate $c_{1}(U)$.}
	\label{FigCU}
\end{figure}
\begin{proof}
	For the "if" part let $A,B$ and $C$ be defined as in lemma B.3. If the value of the control qubit is 0 then $ABC=\mathbb{1}$ is applied to the target qubit. If the value of the control qubit is 1 then $AXBXC=U$ is applied to the target qubit.\\
	\\
	For the "only if" part note that $ABC=\mathbb{1}$ must hold if the output of the circuit is correct when the control qubit is 0. Also if the circuit simulates a $c_{1}(U)$ gate then $AXBXC=U$. Therefore since $\det(AXBXC)=1$, $U$ must also be special unitary and hence $U \in \mathrm{SU}(2)$.
\end{proof}

\section{Stinespring Representation of the Quantum Channel}
When we constructed the unitary operators that implement the quasi extreme channels we make use of the Stinespring representation of the quantum channel $T_{t}$ \cite{stinespring1955positive,wolf2012quantum,nielsen2002quantum}, this appendix shall provide a brief outline of the Stinespring representation as well as an example of its use.

\begin{theorem}
	Let $T: \mathcal{B(H}_{s})\rightarrow \mathcal{B(H}_{s})$ be a CPTP map where $\mathcal{H}_{s}$ is the Hilbert space of the system. Then there exists a Hilbert space $\mathcal{H}_{E}$ called the environment whose and a unitary operator $U$ acting on the joint space $\mathcal{H}_{s} \otimes \mathcal{H}_{E}$
	and a quantum state $\ketbra{0} \in \mathcal{B(H}_{E})$ such that

\begin{align}
		T(\rho)=\mathrm{tr}_{E}[U(\ketbra{0}\otimes \rho)U^{\dagger}], \hspace{3mm} \forall \rho \in \mathcal{B(H}_{s})
	\end{align}
where $\dim(\mathcal{H}_{E})\geq \dim(\mathcal{H}_{s})^{2}$  and the representation is unique up to a unitary equivalence.

\end{theorem}

Given the Kraus representation of the channel $T$ as in equation (\ref{eq3}) one can embed the Kraus in the first block-column of $U$ that acts on the system as well as auxilary qubits that emulate an environment and we populate the rest of block columns so that $U$ is unitary. Now we have,
\begin{align}
	U= \begin{pmatrix}
		K_{1} & \dots & \dots \\
		K_{2} & \dots & \dots \\
		\vdots & \ddots & \vdots\\
		K_{r} & \dots & \dots \\
	\end{pmatrix}
\end{align} 
where $r \leq \dim(\mathcal{H}_{s})^{2}$.  This is done so that when we apply the unitary $U$ to an initial state $\ket{0} \otimes \ket{\psi}$ we get,
\begin{align}
	U(\ket{0} \otimes \ket{\psi})=\sum_{j=1}^{r} \ket{j} \otimes K_{j}\ket{\psi},
\end{align}
where $\{\ket{j}\}_{j=1}^{r}$ is an orthonormal basis for the environment. Now by taking the outer product of this state and tracing out the environment we get the action of the channel $T$ on the initial state $\ketbra{\psi}$ i.e. $T(\ketbra{\psi})$, we can demonstrate this as follows,
\begin{align}
	\mathrm{tr}_{E}\left\{\sum_{j,j'=1}^{r} \ketbra{j}{j'} \otimes K_{j}\ketbra{\psi}K_{j'}^{\dagger}\right\}&=\sum_{j,j'=1}^{r} \mathrm{tr}_{E}\left\{\ketbra{j}{j'}\right\}K_{j}\ketbra{\psi}K_{j'}^{\dagger}\nonumber\\
	&=\sum_{j,j'=1}^{r} \braket{j}{j'}K_{j}\ketbra{\psi}K_{j'}^{\dagger}\nonumber\\
	&=\sum_{j,j'=1}^{r} \delta_{jj'}K_{j}\ketbra{\psi}K_{j'}^{\dagger}\nonumber\\
&=\sum_{j=1}^{r}K_{j}\ketbra{\psi}K_{j}^{\dagger}\nonumber\\
	&=T(\ketbra{\psi}).
\end{align}
Now that we have seen how to use the Stinespring representation of the channel to construct a unitary that acts on the total system plus environment, given the Kraus representation of the channel. Let us see an example of the Stinespring representation for the simple amplitude damping channel.\\
\\
\textbf{Example:} \textit{(Amplitude Damping for a single qubit)} Let’s consider the amplitude damping channel $T_{AD}$, with Kraus operators,
\begin{align}
	K_{0}= \ketbra{0}+ \sqrt{1-p}\ketbra{1}, && K_{1}=\sqrt{p}\ketbra{0}{1}
\end{align}
where in this case $\ket{0}=(1,0)^{T}$ and  $\ket{1}=(0,1)^{T}$. Using the Stinespring representation we get the unitary,
\begin{align}
	U_{AD}=\begin{pmatrix}
		1 & 0 & 0 & 0\\
		0 & \sqrt{1-p} & -\sqrt{p} & 0\\
		0 & \sqrt{p} & \sqrt{1-p} & 0\\
		0 & 0 & 0 & 1\\
	\end{pmatrix},
\end{align}
for which a circuit can be constructed using the methods in appendix B, yielding the circuit in Fig. \ref{FigAD}, with $\theta=2\arcsin(\sqrt{p})$ and $R_{y}(\theta)=\exp(-i\frac{\theta}{2}Y)$.
\begin{figure*}
\centering
	\includegraphics[scale=0.30]{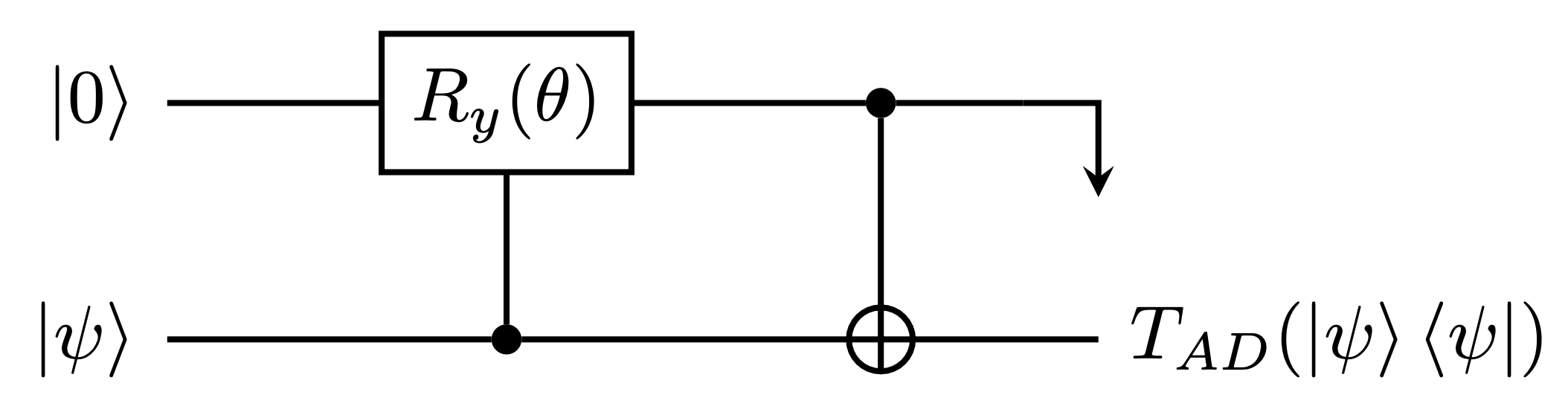}
	\caption{The quantum circuit that implements the amplitude damping channel $T_{AD}$, where the first qubit is the environment in the state $\ket{0}$ and the second qubit is the system. We observe that the arrow on the environment refers to the partial trace taken over the environment.}
	\label{FigAD}
\end{figure*}

\end{document}